%% file: main.tex
\documentclass[a4paper,UKenglish,cleveref, autoref, thm-restate]{lipics-v2021}
\pdfoutput=1 %uncomment to ensure pdflatex processing (mandatatory e.g. to submit to arXiv)
\hideLIPIcs  %uncomment to remove references to LIPIcs series (logo, DOI, ...), e.g. when preparing a pre-final version to be uploaded to arXiv or another public repository

\title{For the Metatheory of Type Theory, Internal Sconing Is Enough} %TODO Please add

\author{Rafaël Bocquet}{Eötvös Loránd University, Budapest, Hungary}{bocquet@inf.elte.hu}{https://orcid.org/0000-0001-6484-9570}
{}
\author{Ambrus Kaposi}{Eötvös Loránd University, Budapest, Hungary}{akaposi@inf.elte.hu}{http://orcid.org/0000-0001-9897-8936}
{Supported by the Bolyai Fellowship of the Hungarian Academy of Sciences and by the ``Application Domain Specific Highly Reliable IT Solutions'' project of the National Research, Development and Innovation Fund of Hungary, financed under the Thematic Excellence Programme TKP2020-NKA-06 funding scheme.}

\author{Christian Sattler}{Chalmers University of Technology, Sweden}{sattler@chalmers.se}{https://orcid.org/0000-0001-6374-4427}
{}

\authorrunning{R. Bocquet and A. Kaposi and C. Sattler} %TODO mandatory. First: Use abbreviated first/middle names. Second (only in severe cases): Use first author plus 'et al.'

\Copyright{Rafaël Bocquet and Ambrus Kaposi and Christian Sattler} %TODO mandatory, please use full first names. LIPIcs license is "CC-BY";  http://creativecommons.org/licenses/by/3.0/


 %TODO mandatory; please add comma-separated list of keywords

\category{} %optional, e.g. invited paper

\relatedversion{} %optional, e.g. full version hosted on arXiv, HAL, or other respository/website
\nolinenumbers %uncomment to disable line numbering

\EventEditors{Marco Gaboardi and Femke van Raamsdonk}
\EventNoEds{2}
\EventLongTitle{8th International Conference on Formal Structures for Computation and Deduction (FSCD 2023)}
\EventShortTitle{FSCD 2023}
\EventAcronym{FSCD}
\EventYear{2023}
\EventDate{July 3--6, 2023}
\EventLocation{Rome, Italy}
\EventLogo{}
\SeriesVolume{260}
\ArticleNo{14}
\newtheorem{construction}[theorem]{Construction}

\usepackage[usenames,dvipsnames]{xcolor}

\usepackage{todonotes}
\presetkeys{todonotes}{inline}{}

\usepackage{tikz}
\usetikzlibrary{cd}
\usetikzlibrary{decorations.pathmorphing}

\usepackage{stmaryrd}
\usepackage{amsmath}
\usepackage{amsthm}
\usepackage{amssymb}
\usepackage{bm}
\usepackage{hyperref}
\usepackage{xspace}
\usepackage{mathpartir}
\usepackage{float}
\usepackage{multirow}

\input{macros.tex}
\input{lmacros.tex}

\begin{document}

\maketitle

\begin{abstract}
Metatheorems about type theories are often proven by interpreting the syntax into models constructed using categorical gluing.
We propose to use only sconing (gluing along a global section functor) instead of general gluing. The sconing is performed internally to a presheaf category, and we recover the original glued model by externalization.
  
Our method relies on constructions involving two notions of models: first-order models (with explicit contexts) and higher-order models (without explicit contexts). Sconing turns a displayed higher-order model into a displayed first-order model.

Using these, we derive specialized induction principles for the syntax of type theory.
The input of such an induction principle is a boilerplate-free description of its motives and methods, not mentioning contexts.
The output is a section with computation rules specified in the same internal language.
We illustrate our framework by proofs of canonicity, normalization and syntactic parametricity for type theory.
\end{abstract}

\input{introduction.tex}

\input{fo_ho_models.tex}
\input{displayed_ho_models.tex}
\input{canonicity.tex}
\input{normalization.tex}
\input{generic_section.tex}

\appendix

\bibliography{main.bib}

\input{proofs.tex}

\input{example_syntactic_parametricity.tex}
\input{example_cubical.tex}

\end{document}

%% file: lmacros.tex
\newcommand{\El}{\mathsf{El}}

\newcommand{\Th}{\CT}

\newcommand{\Tele}{\mathsf{Tele}}
\newcommand{\CTele}{\mathbf{Tele}}
\newcommand{\CScone}{\mathbf{Scone}}

\newcommand{\CSect}{\mathbf{Sect}}

\newcommand{\Ne}{\mathsf{Ne}}
\newcommand{\Nf}{\mathsf{Nf}}
\renewcommand{\ne}{\mathsf{ne}}
\newcommand{\nf}{\mathsf{nf}}
\newcommand{\norm}{\mathsf{norm}}

\newcommand{\CRen}{\mathbf{Ren}}

%% file: introduction.tex
\section{Introduction}\label{sec:introduction}

The syntax of a type theory can be presented as the initial object in the category of models of a generalized algebraic theory (GAT), \ie as a quotient inductive-inductive type (QIIT).
Initiality provides an induction principle for the syntax, namely the dependent eliminator of the QIIT.
Metatheoretic properties of the syntax, such as canonicity or normalization, can then be proven by carefully constructing models of the theory displayed over the syntax, or equivalently the motives and methods of the induction principle.
However, the presentation of the syntax as a QIIT includes an explicit encoding of the substitution calculus of the theory; in particular every type or term former comes with a substitution rule.
If all of the components of a complicated model are written explicitly, one has to prove that they all respect substitution.
More importantly, when working exclusively at this level of generality, it is not easy to abstract the proof methods into reusable theorems.

An alternative to the (first-order) genereralized algebraic presentation is to present a type theory as a second-order or higher-order theory, for example by using a Logical Framework~\cite{FrameworkDefiningLogics} or Uemura's representable map categories~\cite{GeneralFrameworkSemanticsTT}.
A higher-order presentation enables the use of higher-order abstract syntax (HOAS), in which binders are encoded by metatheoretic functions.
In practice, this means that stability under substitution is implicit.
Semantically, HOAS admits an interpretation in the internal language of presheaf categories~\cite{SemAnalysisHOAS}.
Within the internal language of a presheaf category, all constructions are automatically stable under the morphisms of the base category.

In this work we propose an approach that combines the strengths of the first-order and higher-order presentations.
We consider two notions of models of a type theory: first-order models correspond to the first-order presentation, while higher-order models correspond to the higher-order presentation.
Typically, first-order models are categorical models of type theory (categories with families equipped with additional structure), while higher-order models are approximately universes closed under the type formers of the theory.
Both of these notions make sense both externally and in the internal language of a presheaf category.
We also have notions of displayed higher-order and first-order models, corresponding to the motives and methods of induction principles.
We present a small number of constructions switching between external, internal, first-order and higher-order models.
Any of these constructions is individually simple, but they can be composed to derive the induction principles we need to prove metatheoretic results.
The constructions are listed below (FOM = first-order model, HOM = higher-order model). The restriction and externalization operations can also be applied to displayed first-order models and morphisms of first-order models. \\[0.5em]
{\small
  \begin{tabular}{|c|c|c|}
    \hline
    Construction
    & Input
    & Output \\
    \hline
    Internalization
    & FOM $\CM$
    & HOM $\MC$ in $\CPsh(\CM)$ \\
    \hline
    $\CSet$-contextualization
    & HOM $\MM$
    & FOM $\CSet_\MM$ \\
    \hline
    Telescopic contextualization
    & HOM $\MM$
    & FOM $\CTele_\MM$ \\
    \hline
    \multirow{2}{*}{Restriction}
    & Functor $F : \CC \to \CD$,
    & \multirow{2}{*}{FOM $F^\ast(\CM)$ in $\CPsh(\CC)$} \\
    & FOM $\CM$ in $\CPsh(\CD)$ & \\
    \hline
    Externalization
    & FOM $\CM$ in $\CPsh(\CD)$
    & External FOM $1_\CD^\ast(\CM)$ \\
    \hline
    \multirow{2}{*}{$\CScone$-contextualization}
    & FOM $\CM$,
    & \multirow{2}{*}{Displayed FOM $\CScone_{\MM^\bullet}$ over $\CM$} \\
    & Displayed HOM $\MM^\bullet$ over $\CM$ & \\
    \hline
  \end{tabular}
}

The most notable operations are the contextualizations%
\footnote{The word contextualization reflects the fact that these operations make contexts explicit. It is not related to the notion of contextual model.}%
, which turn higher-order models into first-order models.
The $\CSet$-contextualization generalizes the construction of a first-order model from a universe; its underlying category is always the category of sets.
The telescopic contextualization is the contextual core of the $\CSet$-contextualization; it restricts the underlying category to a category of telescopes.
The $\CScone$-contextualization is a generalization of the $\CSet$-contextualization to displayed higher-order and first-order models, which correspond to the motives and methods of induction principles; its underlying category is always the Sierpinski cone, also called scone, of some first-order model $\CM$.

Our main observation is that sconing (that is $\CScone$-contextualization) internally to a presheaf category corresponds when viewed externally to a more complicated gluing construction.
For example, the normalization model from~\cite{CoquandNormalization} can be recovered as the externalization of the $\CScone$-contextualization of a displayed higher-order model constructed in presheaves over the category of renamings.

The main intended application of these constructions is the statement and proofs of relative induction principles.
It is typically the case that the result of induction over the syntax of type theory only holds over some of the syntactic contexts, or is only stable under some of the syntactic substitutions.
For example, canonicity only holds over the empty context.
Normalization holds over every context, but is only stable under renamings.
This situation is described by a functor into the syntax $F : \CR \to \CS$: the result of the induction should be stable under the morphisms of $\CR$.
These functors are called ``figure shapes'' by Sterling~\cite{SterlingThesis}, they are also related to worlds in Twelf~\cite{Twelf}.
We prove the relative induction principles for the following functors. 
\begin{alignat*}{4}
  & 1_\CCat && { }\to{ } && \CS_{\Th} && \quad\text{Canonicity \cite{CoquandNormalization} (Theorem~\ref{thm:relative_induction_principle_terminal})} \\
  & \CRen_\CS && { }\to{ } && \CS_{\Th} && \quad\text{Normalization \cite{altenkirch_et_al:LIPIcs:2016:5972,CoquandNormalization} (Theorem~\ref{thm:relative_induction_principle_renamings})} \\
  & \square && { }\to{ } && \CS_{\mathsf{CTT}} && \quad\text{(Homotopy/strict) canonicity for cubical type theory \cite{coquand2021canonicity} (Theorem~\ref{thm:relative_induction_principle_cubical})} \\
  & \mathcal{A}_{\square} && { }\to{ } && \CS_{\mathsf{CTT}} && \quad\text{Normalization for cubical type theory \cite{NormalizationCTT} (Theorem~\ref{thm:relative_induction_principle_cubical_renaming})}
\end{alignat*}
We also show how to prove canonicity and normalization by instantiating the relative induction principles for $1_\CCat \to \CS_\Th$ and $\CRen_\CS \to \CS_\Th$.
We don't prove canonicity nor normalization for cubical type theory, but we expect that the currently known proofs could be reformulated as instances of the relative induction principles for $\square \to \CS_{\mathsf{CTT}}$ and $\mathcal{A}_{\square} \to \CS_{\mathsf{CTT}}$.

Typically, the category $\CR$ can be described as the initial object of some category of structured categories.
A relative induction principle with respect to the functor $F$ is an induction principle that combines the universal properties of $\CR$ and $\CS$.
In our previous work~\cite{DBLP:journals/corr/abs-2102-11649}, the relative induction principles were stated in terms of the rather ad-hoc notions of ``displayed models without context extensions'' and ``relative sections''.
In the present work, the input of a relative induction principle is a displayed higher-order model $\MS^\bullet$, and the result is just a (first-order) section $\sem{-}$ of $\CScone_{\MS^\bullet}$.
The previous notion of ``displayed model without context extension'' is recovered in the special case of displayed higher-order models over $F^\ast(\CTele_\MS)$.
The following diagram illustrates the constructions involved in the statement and proof of a relative induction principle ($F : \CR \to \CS$). \\[0.5em]
\begin{tikzpicture}
  \node(labelext) {External};
  \node(labelpshs)[right= of labelext] {Internal to $\CPsh(\CS)$};
  \node(labelpshc)[right= of labelpshs] {Internal to $\CPsh(\CR)$};
  \node(labelfo) at (-2,-1.2) {First-order};
  \node(labelho) at (-2,-3) {Higher-order};
  \node(cs) at (0,-1.2) {$\CS$};
  \node(ms) at (2.3,-3) {$\MS$};
  \node(teles) at (2.3,-1.2) {$\CTele_\MS$};
  \node(telec) at (6.6,-1.2) {$F^\ast(\CTele_\MS)$};
  \node(scone) at (9.5,-1.2) {$\CScone_{\MS^\bullet}$};
  \node(dispho) at (9.5,-3) {$\MS^\bullet$};
  \draw[dashed, opacity=0.3] (1.3,0.3) -- (1.3,-3.5);
  \draw[dashed, opacity=0.3] (5.5,0.3) -- (5.5,-3.5);
  \draw[dashed, opacity=0.3] (-3.1,-2.2) -- (10.9,-2.2);
  \draw[|->] (cs) -- (ms) node[near end,left=0pt] { Internalization };
  \draw[|->] (ms) -- (teles) node[midway,right=0pt,align=left,text width=4cm] { {Telescopic\\ contextualization} };
  \draw[|->] (teles) -- (telec) node[midway,above=1pt] { Restriction };
  \draw[|->] (dispho) -- (scone) node[midway,left=0pt,text width=2.5cm,align=right] { {$\CScone$-\\contextualization} };
\end{tikzpicture}

Both $\MS^\bullet$ and $\CScone_{\MS^\bullet}$ are displayed over the internal first-order model $F^\ast(\CTele_\MS)$.
\[ \begin{tikzcd}
    \MS^\bullet
    \ar[rd, -{Triangle[open]}]
    \ar[r, maps to]
    & \CScone_{\MS^\bullet}
    \ar[d, -{Triangle[open]}]
    \\
    & F^\ast(\CTele_\MS)
    \ar[u, bend right, dashed, "{\sem{-}}"']
  \end{tikzcd}
\]
  
An important feature of our work is that the section $\sem{-}$ admits good computational behaviour, although we do not formally analyze this behaviour.
In fact, part of our understanding comes from looking at the computational behaviour of the congruence operation $\ap$ in higher observational type theory~\cite{HOTT-shulman,HOTT-types}, which is also a morphism of (first-order) models internally to presheaves over the syntax of H.O.T.T.

\subsection*{Related work}

\subparagraph*{Logical relations and categorical gluing.}

The initial motivation for this work was the understanding of algebraic and reduction-free normalization proofs for dependent type theories.
Variants of categorical gluing were used to prove canonicity or normalization for simple types~\cite{DBLP:conf/ctcs/AltenkirchHS95,DBLP:conf/ppdp/Fiore02,NormalizationByGluingFreeLT}, System F~\cite{DBLP:conf/lics/AltenkirchHS96}, and dependent types~\cite{altenkirch_et_al:LIPIcs:2016:5972,CoquandNormalization,DBLP:journals/lmcs/CoquandHS22}. These can be contrasted with reduction based normalization proofs such as~\cite{DBLP:journals/pacmpl/0001OV18,pujet:hal-03857705}.

Logical relations were used to prove syntactic parametricity for type theory~\cite{bernardy12parametricity,popl16} and definability for simply typed lambda calculus~\cite{DBLP:conf/tlca/JungT93}. It was shown that categorical gluing generalizes both syntactic parametricity and canonicity proofs~\cite{GluingTT}.

\subparagraph*{Logical frameworks and higher-order abstract syntax.}
Higher-order abstract syntax~(HOAS) is the use of metatheoretic functions to specify syntactic binders.
Hofmann~\cite{SemAnalysisHOAS} has explained how HOAS can be interpreted in the internal language of presheaf categories.

Uemura~\cite{GeneralFrameworkSemanticsTT} has given a general definition of type theory based on these ideas, which we will call second-order generalized algebraic theories (SOGATs).
It generalizes notions of second-order algebraic theories that have been studied by Fiore and Mahmoud~\cite{DBLP:journals/corr/FioreM13}. Harper presents an equational variant of logical framework~\cite{DBLP:journals/corr/abs-2106-01484} for defining theories with bindings corresponding to SOGATs. We believe that the constructions in our paper generalize to any SOGAT.
Gratzer and Sterling \cite{SynCategoriesSketchingAndAdequacy} propose using LCCCs to define higher order theories (without representability conditions) which correspond to our higher order models, but we also consider first order models.

\subparagraph*{Synthetic Tait computability.}

Synthetic Tait computability~(STC, \cite{SterlingThesis,sterling:2022:naive}) is an approach that relies on the internal language of the Artin gluing of toposes, typically constructed by gluing syntactic and semantic toposes.
A pair of open/closed modalities can be used to distinguish the syntactic and semantic parts in the internal language of the Artin gluing.
STC has been applied to proofs of normalization for cubical type theory~\cite{NormalizationCTT}, multimodal type theory~\cite{NormalizationMTT} and simplicial type theory~\cite{weinberger-ahrens-buchholtz-north:2022:types}.

Both STC and our approach provide a synthetic setting for proofs of metatheorems.
Our approach is perhaps simpler in some aspects, \eg we don't use any modalities and don't need to use realignment; we acknowledge that this is partly a matter of preference.
The main advantage of our approach over STC is that we have an internal specification of the result of an induction principle.

In his thesis~\cite{SterlingThesis}, Sterling briefly discusses the notion of Henkin model of a higher-order theory.
A Henkin model is a higher-order model with a non-standard interpretation of the dependent products.
We note that Henkin models are closely related to first-order models: the Henkin models of a second-order theory are equivalent to the democratic first-order models (this is~\cite[Theorem~7.30]{GeneralFrameworkSemanticsTT}).

\subsection*{Contributions}

The main takeaway of this paper is that when using HOAS, we should not discard the first order presentation,
even in an internal setting. In particular, the right notion of displayed
higher-order model (input of an induction principle) lies over a
first-order model. The output of an induction principle is also a
section of first-order models, still internal.

The technical contributions are:
\begin{itemize}
\item the relative induction principles which combine the initiality
  of the underlying category and initiality of the syntax; we derive
  four such induction principles;
\item the internalization, contextualization, externalization
  constructions which let us formulate the relative induction
  principles;
\item the derivation of internal sections by analyzing the category of sections.
\end{itemize}
The main takeaway is supported by applications of the induction principles: boilerplate-free proofs of canonicity, normalization and syntactic parametricity.
To our knowledge, this paper is the first defining the latter in a synthetic setting.
The fact that the notion of section is specified internally has the feature that we can reuse it directly in subsequent inductions.
We exploit this in the proof of uniqueness of normal forms which is proven in a separate step after normalization, relying on how the section computes normal forms.

\subsection*{Structure of the paper}

In Section~\ref{sec:fo_ho_models} we
define first-order- and higher-order models of an example type theory,
and constructions relating them.  In
Section~\ref{sec:displayed_ho_models}, we define displayed higher
order models which collect the motives and methods of an induction
principle. We also define sconing which turns a displayed higher order
model into its first order variant, also providing a notion of
section. Then we move on to applications: we prove canonicity in
Section~\ref{sec:canonicity} using an induction principle relative to
the empty context
(Theorem~\ref{thm:relative_induction_principle_terminal}) which is a
trivial consequence of our previous definitions. In
Section~\ref{sec:normalization}, we prove normalization using an
induction principle relative to renamings
(Theorem~\ref{thm:relative_induction_principle_renamings}). This
induction principle is proved using the methods described in
Section~\ref{sec:generic_section}. Another application (syntactic
parametricity) of
Theorem~\ref{thm:relative_induction_principle_terminal} is described
in Appendix~\ref{sec:parametricity}, and the cubical variants of the
above induction principles are proven in
Appendix~\ref{sec:cubical}.
%(Theorems~\ref{thm:relative_induction_principle_cubical}, \ref{thm:relative_induction_principle_cubical_renaming}).

\subsection*{Background}\label{sec:background}

We assume some familiarity with the categorical semantics of type theory~\cite{CwFs:USD,popl16} and with the use of extensional type theory as the internal language of presheaf categories \cite{Hofmann97syntaxand}.

We use the notion of locally representable dependent presheaf to encode context extensions (see \cref{def:local_rep}).
This definition is a more indexed formulation of the notion of representable natural transformation, which was used by Awodey to give an alternative definition of CwFs, known as natural models~\cite{awodey_2018}.

% Section 7 defines the category of section of an internal disp fo
% model, and the appendix describes how finite limits in C can be lifted
% to the category of sections to obtain the relative induction
% aprinciples.

%%% Local Variables:
%%% mode: latex
%%% TeX-master: "main"
%%% End:

%% file: fo_ho_models.tex
\section{First-order and higher-order models}\label{sec:fo_ho_models}

Our running example is a minimal dependent type theory $\Th$ with only $\Pi$-types, but our constructions directly generalize to larger type theories.
Some other type theories are considered in the appendix, including dependent type theories with universes and cubical type theories.
We leave generalization to arbitrary second-order generalized algebraic theories to future work.

We define notions of higher-order and first-order models for $\Th$.
A higher-order model is essentially a universe closed under dependent products, while a first-order model is a category with families (CwF) equipped with $\Pi$-types.
The higher-order models are the models of some higher-order theory $\Th^{\mathsf{ho}}$ (a theory whose operations can have a higher-order sort; classified by some locally cartesian closed category), whereas the first-order models are the models of some (first-order) essentially%
\footnote{The distinction between generalized algebraic theories and essentially algebraic theories is not relevant in this paper.}
algebraic theory $\Th^{\mathsf{fo}}$ (a theory whose operations have a first-order sort; classified by some finitely complete category).

\subsection{Definitions}

\begin{definition}
  A \defemph{higher-order model} of $\Th$ consists of the following families and operations:
  \begingroup{}\allowdisplaybreaks{}
  \begin{alignat*}{2}
    & \Ty && :
      \SSet, \\
    & \Tm && :
      \Ty \to \SSet, \\
    & \Pi && :
      \forall (A : \Ty) (B : \Tm(A) \to \Ty) \to \Ty, \\
    & \app && :
      \forall A\ B\ (f : \Tm(\Pi(A,B)))\ (a : \Tm(A)) \to \Tm(B(a)), \\
    & \lam && :
      \forall A\ B\ (b : (a : \Tm(A)) \to \Tm(B(a))) \to \Tm(\Pi(A,B)), 
    % & \BoolTy :
    %   \Ty, \\
    % & \true :
    %   \Tm(\BoolTy), \\
    % & \false :
    %   \Tm(\BoolTy), \\
    % & \elimBool :
    %   \forall (P : \Tm(\BoolTy) \to \Ty)\ (t : \Tm(P(\true)))\ (f : \Tm(P(\false)))\ (b : \Tm(\BoolTy)) \\
    % & \phantom{\elimBool}\quad\to \Tm(P(b)),
  \end{alignat*}\endgroup{}
  subject to equations corresponding to the $\beta$- and $\eta$-rules: \\
  \begin{minipage}{0.5\textwidth}
  \[
  \lam(\lambda a \mapsto \app(f,a)) = f,
  \]
  \end{minipage}
  \begin{minipage}{0.5\textwidth}
  \[
  \app(\lam(b),a) = b(a). \tag*{\lipicsEnd{}}
  \]
  \end{minipage}
\end{definition}

% \begin{definition}
%   A \defemph{second-order model} of $\Th$ in a presheaf category $\CPsh(\CC)$ is an internal higher-order model of $\Th$ in $\CPsh(\CC)$, such the family $\Tm : \Ty \to \SSet$ is a locally representable presheaf family.
%   \lipicsEnd{}
% \end{definition}

\begin{definition}
  A \defemph{first-order model} of $\Th$ is a CwF $\CC$ equipped with $\Pi$-types. 
  \lipicsEnd{}
\end{definition}

The notion of first-order model can be presented by a first-order generalized algebraic theory $\Th^{\mathsf{fo}}$.
Equivalently, a first-order model is a category $\CC$ with a terminal object $1_\CC$, along with a (global) higher-order model $\MC$ of $\Th$ in $\CPsh(\CC)$ such that the dependent presheaf $\MC.\Tm$ is locally representable.
The higher-order model $\MC$ is called the internalization of $\CC$.

We use blackboard bold letters $\MM$, $\MN$, $\MC$, $\MS$, $\MR$, \etc to refer to higher-order models, and calligraphic letters $\CM$, $\CN$, $\CC$, $\CS$, $\CR$, \etc to refer to first-order models.
We try to use the corresponding letter for the underlying internal higher-order model of a first-order model, \eg if $\CC$ is an external first-order model, we use $\MC$ for its underlying internal higher-order model in $\CPsh(\CC)$.
We denote the components of a model by $\MC.\Ty$, $\CC.\Tm$, $\CC.\Pi$, \etc.

We write $\CMod_\Th$ for the $1$-category of first-order models of $\Th$, and $\CS_\Th$ or just $\CS$ for its initial object (the letter ``S'' standing for both syntax and substitutions).

\subsection{Contextualization}

In general, we almost never want to construct all of the components of a first-order model explicitly, because checking functoriality and naturality conditions without relying on the internal language of a presheaf model is tedious.
For some models however, the functoriality and naturality conditions hold trivially.
This is the case for the ``standard model'' of type theory over the category of sets: when defining this standard model in intensional type theory, all naturality and functoriality conditions hold definitionally.

The construction of the standard model generalizes to the construction of a first-order model from a higher-order model, which we now describe.

\begin{construction}[$\CSet$-Contextualization]\label{constr:contextualization}
  Let $\MM$ be a higher-order model of $\Th$.
  We construct a first-order model $\CSet_\MM$, called the \defemph{$\CSet$-contextualization} of $\MM$.
  Its underlying category is the category $\CSet$ of sets.

  The types and terms of $\CSet_\MM$ are indexed families of types and terms of $\MM$:
  \begin{itemize}
  \item A type over $\Gamma \in \CSet$ is a function $\Gamma \to \MM.\Ty$.
  \item Type substitution along a function $f : \Delta \to \Gamma$ is precomposition with $f$.
  \item A term of type $A : \Gamma \to \MM.\Ty$ is a dependent function $(\gamma : \Gamma) \to \MM.\Tm(A(\gamma))$.
  \item Term substitution along a function $f : \Delta \to \Gamma$ is precomposition with $f$.
  \item The functoriality of substitution is associativity of function composition.
  \end{itemize}

  The context extensions are given by dependent sums in $\CSet$:
  \begin{alignat*}{3}
    & (\Gamma.A) && \triangleq{ } && (\gamma : \Gamma) \times \MM.\Tm(A(\gamma)).
  \end{alignat*}

  The type-theoretic operations are all defined pointwise:
  \begin{alignat*}{3}
    & \CSet_\MM.\Pi(\Gamma,A,B) && \triangleq{ }
    && \lambda \gamma \mapsto \MM.\Pi(A(\gamma), \lambda a \mapsto B(\gamma,a)), \\
    & \CSet_\MM.\app(\Gamma,f,a) && \triangleq{ }
    && \lambda \gamma \mapsto \MM.\app(f(\gamma), a(\gamma)), \\
    & \CSet_\MM.\lam(\Gamma,b) && \triangleq{ }
    && \lambda \gamma \mapsto \MM.\lam(\lambda a \mapsto b(\gamma,a)).
    % & \CSet_\MM.\BoolTy(\Gamma) && \triangleq{ }
    % && \lambda \gamma \mapsto \MM.\BoolTy, \\
    % & \CSet_\MM.\true(\Gamma) && \triangleq{ }
    % && \lambda \gamma \mapsto \MM.\true, \\
    % & \CSet_\MM.\false(\Gamma) && \triangleq{ }
    % && \lambda \gamma \mapsto \MM.\false, \\
    % & \CSet_\MM.\elimBool(\Gamma,P,t,f,b) && \triangleq{ }
    % && \lambda \gamma \mapsto \MM.\elimBool(\lambda b \mapsto P(\gamma,b), t(\gamma), f(\gamma), b(\gamma)).
  \end{alignat*}
  The $\beta$- and $\eta$-rules for $\CSet_\MM$ hold as a consequence of the corresponding rules for $\MM$.

  The naturality conditions are all trivial.
  For example, in the case of the $\Pi$ type former, we have to check
  $\Pi(\Gamma,A,B)[f] = \Pi(\Delta,A[f],B[f^+])$
  for any $f : \Delta \to \Gamma$, where $f^+(\gamma,a) = (f(\gamma),a)$.
  This amounts to checking the equality
  {\small \[ (\lambda \gamma \mapsto \MM.\Pi(A(\gamma), \lambda a \mapsto B(\gamma,a))) \circ f = (\lambda \gamma \mapsto \MM.\Pi((A \circ f)(\gamma), \lambda a \mapsto (B \circ f^+)(\gamma,a))).
      \lipicsEnd{}
  \] }
\end{construction}

\begin{remark}
  Note that the underlying category $\CSet$ of $\CSet_\MM$ now has two CwF structures:
  \begin{itemize}
  \item An \emph{inner} CwF structure, as defined in~\cref{constr:contextualization}.
  \item An \emph{outer} CwF structure, corresponding to the usual CwF structure on the category of sets, modeling extensional type theory.
  \end{itemize}
  Together, they form a model of two-level type theory~\cite{TwoLevelTypeTheoryAndApplications}.
\end{remark}

\subsection{Telescopic contextualization}

A first-order order model is said to be contextual when every object can be uniquely written as an iterated context extension starting from the empty context.
The contextual first-order models form a coreflective subcategory $\CMod_\Th^\cxl$ of $\CMod_\Th$: the inclusion $\CMod_\Th^\cxl \to \CMod_\Th$ has a right adjoint $\cxl$: the \emph{contextual core} $\cxl(\CC)$ has as objects the iterated context extensions of $\CC$, also known as telescopes, over the empty context.

\begin{definition}
  Let $\MM$ be a higher-order model of $\Th$.
  The \defemph{telescopic contextualization} $\CTele_\MM$ is the contextual core of the $\CSet$-contextualization $\CSet_\MM$.
  \lipicsEnd{}
\end{definition}

By general properties of the contextual core, there is a model morphism $\dquote{-} : \CTele_\MM \to \CSet_\MM$ that is bijective on types and terms.
(There is a cofibrantly generated factorization system on first-order models with such morphisms as its right class.
The contextual models are precisely those in the left class.)

When working internally to some presheaf category $\CPsh(\CC)$, another related construction involves the internal subcategory spanned by $\yo : \Ob_\CC \to \SSet$, where $\Ob_\CC$ is the discrete presheaf on the set of objects of $\CC$, and $\yo$ internalizes the Yoneda embedding.
This ``Yoneda universe'' has been used by Hu \etal{}~\cite{CategoryTheoreticViewContextualTypes} to give semantics to contextual types.
% The $\CSet$-contextualization and the telescopic contextualization also give slightly different semantics to contextual types.

\subsection{Internal first-order models}

Since the notion of first-order model is described by an essentially algebraic theory, it can be interpreted in any finitely complete category.
In particular, there is a notion of internal first-order model in any category $\widehat{\CC}$ of small presheaves, obtained by letting $\SSet$ stand for the Hofmann-Streicher universe of the presheaf topos $\CPsh(\CC)$ in the definition of model.

\begin{proposition}\label{prop:eqv_definitions_fom}
  The following three notions are equivalent:
  \begin{enumerate}
  \item\label{itm:internal_fo_model_1} First-order models of $\Th$ in $\CPsh(\CC)$;
  \item\label{itm:internal_fo_model_2} Finite limit preserving functors $\Th^{\mathsf{fo}} \to \widehat{\CC}$, where $\Th^{\mathsf{fo}}$ is the finitely complete category classifying the first-order models of $\Th$;
  \item\label{itm:internal_fo_model_3} Functors $\CC \to \CMod_\Th^\op$.
  \end{enumerate}
\end{proposition}
\begin{proof}
  This is well-known~\cite[D1.2.14]{Johnstone:592033}.
  The equivalence between (\ref{itm:internal_fo_model_1}) and (\ref{itm:internal_fo_model_2}) is the fact that $\Th^{\mathsf{fo}}$ classifies the first-order models of $\Th$.
  The equivalence between (\ref{itm:internal_fo_model_2}) and (\ref{itm:internal_fo_model_3}) follows from the fact that finite limits in $\widehat{\CC}$ are computed pointwise.
\end{proof}

\subsection{Restriction and externalization}

Another important operation on first-order models is the restriction of a first-order model $\CM$ internal to $\CPsh(\CD)$ along a functor $F : \CC \to \CD$.
This restricted model $F^{\ast}(\CM)$ is a first-order model internal to $\CPsh(\CC)$.
If $\CM$ is seen as a functor $\CD \to \CMod_\Th^\op$, then the restriction $F^\ast(\CM) : \CC \to \CMod_\Th^\op$ is simply the precomposition $(\CM \circ F)$.
If $\CM$ is seen instead as a finite-limits preserving functor $\Th^{\mathsf{fo}} \to \CPsh(\CD)$, then $F^\ast(\CM)$ is postcomposition with the inverse image functor $F^\ast : \CPsh(\CD) \to \CPsh(\CC)$.
These two definitions coincide up to the equivalence of~\cref{prop:eqv_definitions_fom}; which is thus natural in the base category.

\begin{remark}
  A more explicit computation of the restriction can be given in the internal language of $\CPsh(\CC)$ using the dependent right adjoint associated to the adjunction $(F_! \dashv F^\ast)$.
  When $\CM$ is the $\CSet$- or telescopic contextualization of a higher-order model, then the dependent right adjoint allows for the use of HOAS when working with $F^\ast(\CM)$.
  \lipicsEnd{}
\end{remark}

A special case of the restriction is the \emph{externalization} of an internal first-order model.
\begin{definition}
  Let $\CC$ be any category with a terminal object $1_\CC$, and consider the functor $1_\CC : 1_\CCat \to \CC$ that selects this terminal object.
  For any internal first-order model $\CM$ in $\CPsh(\CC)$, we have an external first-order model $1_\CC^\ast(\CM)$, called the \defemph{externalization} of $\CM$.
  \lipicsEnd{}
\end{definition}

Given any higher-order model $\MM$ in $\CPsh(\CC)$, we can construct the externalization $1_{\CC}^{\ast}(\CTele_\MM)$ of its telescopic first-order model.
Up to isomorphism, all external contextual first-order models arise as the externalization of a telescopic contextualization.
\begin{lemma}
  Let $\CC$ be an external first-order model, with $\MC$ its underlying internal higher-order model.
  Then $1_\CC^\ast(\Tele_\MC)$ is the contextual core of $\CC$.
\end{lemma}
\begin{proof}
  See~\cref{lem:contextual_slice_terminal}.
\end{proof}

In particular, since the initial model $\CS$ is contextual, the externalization $1_\CS^\ast(\Tele_\MS)$ of its telescopic contextualization is isomorphic to $\CS$.

We can also construct the externalization $1_\CC^{\ast}(\CSet_\MC)$ of the internal $\CSet$-contextualization of an higher-order model $\MC$.
The underlying category of $1_\CC^\ast(\CSet_\MC)$ is the category of presheaves over $\CC$ (restricted to some universe level); and $1_\CC^\ast(\CSet_\MC)$ is an external model of two-level type theory (its underlying category is the restriction of $\CPsh(\CC)$ to some universe level).
Recall that there is, internally to $\CPsh(\CC)$, a morphism $\dquote{-} : \CTele_\MC \to \CSet_\MC$ of first-order models.
This morphism can also be externalized, giving an external morphism $1_\CC^\ast(\dquote{-}) : 1_\CC^\ast(\Tele_\MC) \to 1_\CC^\ast(\CSet_\MC)$ of first-order models.
When $1_\CC^\ast(\Tele_\MC) \cong \CC$, this is a simple construction of the embedding of $\CC$ into the presheaf model of two-level type theory.

%%% Local Variables:
%%% mode: latex
%%% TeX-master: "main"
%%% End:

%% file: displayed_ho_models.tex
\section{Displayed higher-order models}\label{sec:displayed_ho_models}

\subsection{Motives and methods}

We now define the notion of displayed higher-order model, which collects the motives and methods of induction principles.
One could expect that a displayed higher-order model would be displayed over a base higher-order model.
We instead define the notion of displayed higher-order model over a base first-order model; it is always possible to turn higher-order models into first-order models using a contextualization, but not every first-order model arises in this way.

\begin{definition}
  Let $\CM$ be a first-order model of $\Th$.
  A \defemph{displayed higher-order model} $\MM^\bullet$ over $\CM$ consists of the following data:
  \begingroup{}\allowdisplaybreaks
  \begin{alignat*}{3}
    & \Ty^{\bullet} && :{ }
    && \CM.\Ty(1_\CM) \to \SSet, \\
    & \Tm^{\bullet} && :{ }
    && \forall (A : \CM.\Ty(1_\CM))\ (A^{\bullet} : \Ty^{\bullet}(A)) \to \CM.\Tm(1_\CM,A) \to \SSet, \\
    & \Pi^{\bullet} && :{ }
    && \forall (A : \CM.\Ty(1_\CM))\ (A^{\bullet} : \Ty^{\bullet}(A)) \\
    &&&&& \phantom{\forall} (B : \CM.\Ty(1_\CM.A))\ (B^{\bullet} : \forall (a : \CM.\Tm(1_\CM,A))(a^{\bullet} : \Tm^{\bullet}(A^{\bullet},a)) \to \Ty^{\bullet}(B[a])) \\
    &&&&& { }\to \Ty^{\bullet}(\Pi(A,B)), \\
    & \app^\bullet && :{ }
    && \forall A\ A^\bullet\ B\ B^\bullet\ (f : \CM.\Tm(1_\CM,\Pi(A,B)))\ (f^\bullet : \Tm^\bullet(\Pi^{\bullet}(A^{\bullet},B^{\bullet}), f)) \\
    &&&&& \phantom{\forall} (a : \CM.\Tm(1_\CM,A))\ (a^{\bullet} : \Tm^{\bullet}(A^{\bullet},a)) \\
    &&&&& { }\to \Tm^{\bullet}(B^{\bullet}(a^{\bullet}),\app(f,a)), \\
    & \lam^\bullet && :{ }
    && \forall A\ A^\bullet\ B\ B^\bullet\ (b : \CM.\Tm(1_\CM.(a:A),B[a])) \\
    &&&&& \phantom{\forall} (b^\bullet : \forall (a : \CM.\Tm(1_\CM,A))(a^{\bullet} : \Tm^{\bullet}(A^{\bullet},a)) \to \Tm^{\bullet}(B^{\bullet}(a^{\bullet}),b[a])) \\
    &&&&& { } \to \Tm^\bullet(\Pi^{\bullet}(A^{\bullet},B^{\bullet}),\lam(b)),
    % & \BoolTy^{\bullet} && :{ }
    % && \Ty^\bullet(\BoolTy), \\
    % & \true^{\bullet} && :{ }
    % && \Tm^\bullet(\BoolTy^{\bullet},\true), \\
    % & \false^{\bullet} && :{ }
    % && \Tm^\bullet(\BoolTy^{\bullet},\false), \\
    % & \elimBool^{\bullet} && :{ }
    % && \forall P\ (P^\bullet : \forall (b : \CM.\Tm(1_\CM,\BoolTy)) (b^\bullet : \Tm^\bullet(\BoolTy^\bullet,b)) \to \Ty^\bullet(P[b])) \\
    % &&&&& \phantom{\forall} t\ (t^\bullet : \Tm^\bullet(P^\bullet(\true^\bullet,t)))\ f\ (f^\bullet : \Tm^\bullet(P^\bullet(\false^\bullet,f)))\ b\ (b^\bullet : \Tm^\bullet(\BoolTy^\bullet,b)) \\
    % &&&&& { }\to \Tm^{\bullet}(P^\bullet(b^\bullet),\elimBool(P,t,f,b)),
  \end{alignat*}\endgroup{}
  such that the following equalities hold: \\
  \begin{minipage}{0.5\textwidth}
  \[
  \app^\bullet(\lam^\bullet(b^\bullet),a^\bullet) = b^\bullet(a^\bullet),
  \]
  \end{minipage}
  \begin{minipage}{0.5\textwidth}
  \[
  \lam^\bullet(\lambda a^\bullet \mapsto \app^\bullet(f^\bullet,a^\bullet)) = f^\bullet. \tag*{\lipicsEnd{}}
  \]
  \end{minipage}
\end{definition}

Most of the components of a displayed higher-order model only depend on the closed types and terms of $\CM$; only the binders need to refer to open types and terms.

Note that the data of a displayed higher-order model over the terminal first-order model is equivalent to the data of a non-displayed higher-order model.

\subsection{Displayed contextualization}

Given any displayed higher-order model $\MM^\bullet$ over $\CM$, we construct a displayed first-order model over $\CM$.
This construction is a displayed generalization of the $\CSet$-contextualization.

The underlying displayed category of this construction is the \emph{Sierpinski cone}, or scone, of $\CM$.
The scone of a category $\CC$ with a terminal object is the comma category $(\CSet \downarrow \Gamma_\CC)$, where $\Gamma_\CC : \CC \to \CSet$ is the global section functor $\Gamma_\CC = \CC(1_\CC, -)$.

\begin{construction}[Displayed contextualization]\label{constr:internal_gluing_model}
  Fix a displayed higher-order model $\MM^\bullet$ over a first-order model $\CM$.
  We construct a displayed first-order model $\CScone_{\MM^\bullet}$ over $\CM$, called the \defemph{displayed contextualization} of $\MM^\bullet$.
  \begin{itemize}
  \item An object of $\CScone_{\MM^\bullet}$ over $\Gamma \in \CM$ is a family
    \[ \Gamma^\dagger : \CM(1_\CM, \Gamma) \to \SSet \]
    over the global elements (\ie closing substitutions) of $\Gamma$.
  \item A morphism of $\CScone_{\MM^\bullet}$ from $\Gamma^\dagger$ to $\Delta^\dagger$ over a base morphism $f : \CM(\Gamma, \Delta)$ is a family
    \[ f^\dagger : \forall (\gamma : \CM(1_\CM,\Gamma)) \to \Gamma^\dagger(\gamma) \to \Delta^\dagger(f \circ \gamma). \]

    The identity displayed morphism is given by
    \[ \id^\dagger \triangleq \lambda \gamma\ \gamma^\bullet \mapsto \gamma^\bullet, \]
    whereas the composition of two displayed morphisms $f^\dagger$ and $g^\dagger$ is
    \[ f^\dagger \circ^\dagger g^\dagger \triangleq \lambda \gamma\ \gamma^\bullet \mapsto f^\dagger(g^\dagger(\gamma^\bullet)). \]

  \item A type of $\CScone_{\MM^\bullet}$ over an object $\Gamma^\dagger$ and a type $A : \CM.\Ty(\Gamma)$ is a function
    \[ A^\dagger : \forall \gamma\ (\gamma^\dagger : \Gamma^\dagger(\gamma)) \to \Ty^\bullet(A[\gamma]). \]

    The restriction of $A^\dagger$ along a displayed morphism $f^\dagger$ is
    \[ A^\dagger[f^\dagger] \triangleq \lambda \gamma^\bullet \mapsto A^\dagger(f^\dagger(\gamma^\bullet)). \]

  \item A term of $\CScone_{\MM^\bullet}$ of type $A^\dagger$ over an object $\Gamma^\dagger$ and a term $a : \CM.\Tm(\Gamma,A)$ is a function
    \[ a^\dagger : \forall \gamma\ (\gamma^{\bullet} : \Gamma^\dagger(\gamma)) \to \Tm^\bullet(A^\dagger(\gamma^\bullet), a[\gamma]). \]

    The restriction of $a^\dagger$ along a displayed morphism $f^\dagger$ is
    \[ a^\dagger[f^\dagger] \triangleq \lambda \gamma^\bullet \mapsto a^\dagger(f^\dagger(\gamma^\bullet)). \]

  \item The empty displayed context is the family
    \[ 1^\dagger \triangleq \lambda \_ \mapsto \Unit. \]
  \item The extension of a displayed context $\Gamma^\dagger$ by a displayed type $A^\dagger$ is the family
    \[ (\Gamma^\dagger.A^\dagger) \triangleq \lambda \angles{\gamma,a} \mapsto (\gamma^\bullet : \Gamma^\dagger(\gamma)) \times A^\dagger(\gamma^\bullet,a). \]

  \item All type- and term- formers are defined pointwise using the corresponding component of $\MM^\bullet$:
    \begingroup{}\allowdisplaybreaks
    \begin{alignat*}{3}
      & \Pi^\dagger(A^\dagger,B^\dagger) && \triangleq{ }
      && \lambda \gamma^\bullet \mapsto \Pi^\bullet(A^\dagger(\gamma^\bullet), \lambda a^\bullet \mapsto B^\dagger(\gamma^\bullet,a^\bullet)), \\
      & \app^\dagger(f^\dagger,a^\dagger) && \triangleq{ }
      && \lambda \gamma^\bullet \mapsto \app^\bullet(f^\dagger(\gamma^\bullet),a^\dagger(\gamma^\bullet)), \\
      & \lam^\dagger(b^\dagger) && \triangleq{ }
      && \lambda \gamma^\bullet \mapsto \lam^\bullet(\lambda a^\bullet \mapsto b^\dagger(\gamma^\bullet,a^\bullet)).
      % & \BoolTy^\dagger && \triangleq{ }
      % && \lambda \gamma^\bullet \mapsto \BoolTy^\bullet, \\
      % & \true^\dagger && \triangleq{ }
      % && \lambda \gamma^\bullet \mapsto \true^\bullet, \\
      % & \false^\dagger && \triangleq{ }
      % && \lambda \gamma^\bullet \mapsto \false^\bullet, \\
      % & \elimBool^\dagger(P^\dagger,t^\dagger,f^\dagger,b^\dagger) && \triangleq{ }
      % && \lambda \gamma^\bullet \mapsto \elimBool^\bullet(\lambda b^\bullet \mapsto P^\dagger(\gamma^\bullet,b^\bullet), t^\dagger(\gamma^\bullet), f^\dagger(\gamma^\bullet), b^\dagger(\gamma^\bullet)).
    \end{alignat*}\endgroup{}
  \item The $\beta$- and $\eta$-rules hold as a consequence of the $\beta$- and $\eta$-rules of $\MM^\bullet$.
  \item All naturality conditions are trivial.
    \lipicsEnd{}
  \end{itemize}
\end{construction}

Note that when $\MM$ is a higher-order model seen as a displayed higher-order model over the terminal first-order model, then $\CScone_{\MM}$ is equivalent to $\CSet_{\MM}$.

\subsection{Sections of a displayed higher-order model}

The notion of displayed higher-order model corresponds to the motives and methods of an induction principle.
We now define the notion of section of a displayed higher-order model, corresponding to the result of applying an induction principle: it is simply defined as a section of the displayed contextualization.

\begin{definition}
  A \defemph{section} of a displayed higher-order model $\MM^\bullet$ is a section $\sem{-}$ of its displayed contextualization $\CScone_{\MM^\bullet}$ (in $\CMod_\Th$).
  \lipicsEnd{}
\end{definition}

The definition of section of a displayed higher-order model $\MM^\bullet$ over $\CM$ can be unfolded to the following components:
\begin{itemize}
\item For every object $\Gamma : \CM$, a family
  \[ \sem{\Gamma} : \CM(1_\CM,\Gamma) \to \CSet \]
  of \emph{environments}.
\item For every morphism $f : \CM(\Gamma,\Delta)$, a family
  \[ \sem{f} : \forall \gamma \to \sem{\Gamma}(\gamma) \to \sem{\Delta}(f \circ \gamma) \]
  of maps between environments.
\item For every type $A : \CM.\Ty(\Gamma)$, a family
  \[ \sem{A} : \forall \gamma\ (\gamma^\bullet : \sem{\Gamma}(\gamma)) \to \Ty^\bullet(A[\gamma]) \]
  of displayed types over closures of $A$.
\item For every term $a : \CM.\Tm(\Gamma,A)$, a family
  \[ \sem{a} : \forall \gamma\ (\gamma^\bullet : \sem{\Gamma}(\gamma)) \to \Tm^\bullet(\sem{A}(\gamma^\bullet), a[\gamma]) \]
  of displayed terms over closures of $a$.
\item Subject to functoriality and naturality equations:
  \begin{alignat*}{1}
    & \sem{\id}(\gamma^\bullet) = \gamma^\bullet, \\
    & \sem{f \circ g}(\gamma^\bullet) = \sem{f}(\sem{g}(\gamma^\bullet)), \\
    & \sem{A[f]}(\gamma^\bullet) = \sem{A}(\sem{f}(\gamma^\bullet)), \\
    & \sem{a[f]}(\gamma^\bullet) = \sem{A}(\sem{f}(\gamma^\bullet)).
  \end{alignat*}
\item Such that context extensions are preserved:
  \begin{alignat*}{1}
    & \sem{1_\CM}(\star) = \{\star\},
    \\
    & \sem{\Gamma.A}(\gamma,a) = (\gamma^\bullet : \sem{\Gamma}(\gamma)) \times (a^\bullet : \Tm^\bullet(\sem{A}(\gamma^\bullet), a)), \\
    & \sem{\lambda \gamma \mapsto (\delta(\gamma),a(\gamma))}(\gamma) = (\sem{\delta}(\gamma), \sem{a}(\gamma)).
  \end{alignat*}
\item With computation rules for every type and term former:
  \begin{alignat*}{1}
    & \sem{\lambda \gamma \mapsto \Pi(A(\gamma),\lambda a \mapsto B(\gamma,a))}(\gamma^\bullet) = \Pi^\bullet(\sem{A}(\gamma^\bullet), \lambda a^\bullet \mapsto \sem{B}(\gamma^\bullet)), \\
    & \sem{\lambda \gamma \mapsto \app(f(\gamma), a(\gamma))}(\gamma^\bullet) = \app^\bullet(\sem{f}(\gamma^\bullet), \sem{a}(\gamma^\bullet)), \\
    & \sem{\lambda \gamma \mapsto \lam(b(\gamma))}(\gamma^\bullet) = \lam^\bullet(\lambda a^\bullet \mapsto \sem{b}(\gamma^\bullet,a^\bullet)).
    % & \sem{\lambda \gamma \mapsto \BoolTy}(\gamma^\bullet) = \BoolTy^\bullet, \\
    % & \sem{\lambda \gamma \mapsto \true}(\gamma^\bullet) = \true^\bullet, \\
    % & \sem{\lambda \gamma \mapsto \false}(\gamma^\bullet) = \false^\bullet, \\
    % & \sem{\lambda \gamma \mapsto \elimBool(\lambda a \mapsto P(\gamma,a), t(\gamma), f(\gamma), b(\gamma))}(\gamma^\bullet) \\
    % & \quad = \elimBool^\bullet(\lambda a^\bullet \mapsto \sem{P}(\gamma^\bullet,a^\bullet), \sem{t}(\gamma^\bullet), \sem{f}(\gamma^\bullet), \sem{b}(\gamma^\bullet)).
  \end{alignat*}
\end{itemize}

When $x$ is a closed type or term of $\CM$, we write $\sem{x}$ for the interpretation $\sem{x}(\star)$ of $x$ in the empty environment.
We may use underlined names to distinguish the variable of open terms. For instance, we may write $\sem{\app(\underline{f},\underline{a})}[\underline{f} \mapsto f', \underline{a} \mapsto a']$ instead of $\sem{\lambda (f,a) \mapsto \app(f,a)}(f',a')$.

\begin{remark}
  Let $\MS^\bullet$ be a displayed higher-order model over the first-order model $F^\ast(\Tele_\MS)$ in $\CPsh(\CC)$, for some functor $F : \CC \to \CS$.
  The displayed contextualization $\CScone_{\MS^\bullet}$ is a displayed first-order model over $F^\ast(\Tele_\MS)$.
  Then its externalization $1^{\ast}_{\CC}(\CScone_{\MS^\bullet})$ is an external displayed first-order model lying over $1^\ast_\CC(F^\ast(\Tele_\MS)) = 1^\ast_\CS(\Tele_\MS)$.
  Up to the isomorphism $1^\ast_\CS(\Tele_\MS) \cong \CS$, the externalized $\CScone$-contextualization $1^{\ast}_{\CC}(\CScone_{\MS^\bullet})$ coincides with gluing.
  Its underlying category is the comma category $(\CS \downarrow N_F)$, where $N_F : \CS \to \CPsh(\CC)$ is the nerve functor $\CS \xrightarrow{\yo} \CPsh(\CS) \xrightarrow{F^\ast} \CPsh(\CC)$.
  \lipicsEnd{}
\end{remark}

%%% Local Variables:
%%% mode: latex
%%% TeX-master: "main"
%%% End:

%% file: canonicity.tex
\section{Example: canonicity proof}\label{sec:canonicity}

As a first example of a relative induction principle and its application, we prove canonicity for $\Th$ extended with booleans (given by a type former $\BoolTy$ with constructors $\true$ and $\false$ and a dependent eliminator $\elimBool$ with two computation rules).

We use the induction principle relative to the functor $1_\CS : 1_\CCat \to \CS$ that selects the terminal object in the syntax $\CS$.
It turns out that proving this specific relative induction principle is trivial.

\begin{theorem}[Induction principle for $\MS$ relative to $1_\CS : 1_\CCat \to \CS$]\label{thm:relative_induction_principle_terminal} \mbox{}\\
  Let $\MS^{\bullet}$ be a displayed higher-order model over the initial model $\CS$, or equivalently over $1_\CS^\ast(\Tele_\MS)$.
  Then $\CScone_{\MS^{\bullet}}$ admits a section $\sem{-}$ over $\CS$.
\end{theorem}
\begin{proof}
  By initiality of~$\CS$.
\end{proof}

We now construct the displayed higher-order model $\MS^{\bullet}$ over $1_\CS^\ast(\Tele_\MS)$ that will be used to prove canonicity.
A displayed type $A^{\bullet}$ over a closed type $A : \CS.\Ty(1_\CS)$ is a $\SSet$-valued logical predicate over the closed terms of type $A$:
\[ \Ty^{\bullet}(A) \triangleq \CS.\Tm(1_\CS,A) \to \SSet. \]
A displayed term $a^{\bullet}$ of type $A^\bullet$ over a closed term $a : \CS.\Tm(1_\CS,A)$ is an element of the logical predicate $A^{\bullet}$ evaluated at $a$:
\[ \Tm^{\bullet}(A^{\bullet}, a) \triangleq A^{\bullet}(a). \]
Given logical predicates $A^\bullet$ and $B^\bullet$, the logical predicate $\Pi^\bullet(A^\bullet,B^\bullet)$ expresses the fact that functions $f : \CS.\Tm(1_\CS,\Pi(A,B))$ should preserve the logical predicates.
\begin{alignat*}{3}
  & \Pi^{\bullet}(A^\bullet, B^\bullet) && \triangleq{ }
  && \lambda (f : \CS.\Tm(1_\CS,\Pi(A,B))) \mapsto (\forall a\ a^\bullet \to B^\bullet(a^\bullet, \app(f,a))), \\
  & \app^{\bullet}(f^\bullet, a^\bullet) && \triangleq{ }
  && f^\bullet(a^\bullet), \\
  & \lam^{\bullet}(b^\bullet) && \triangleq{ }
  && \lambda a^\bullet \mapsto b^\bullet(a^\bullet).
\end{alignat*}
It is easy to check that the displayed $\beta$- and $\eta$-rules hold.
The logical predicate $\BoolTy^\bullet : \CS.\Tm(1_\CS,\BoolTy) \to \SSet$ is defined as an inductive family with two constructors $\true^\bullet : \BoolTy^\bullet(\true)$ and $\false^\bullet : \BoolTy^\bullet(\false)$.
% \begin{alignat*}{3}
%   & \true^\bullet && :{ }
%   && \BoolTy^\bullet(\true), \\
%   & \false^\bullet && :{ }
%   && \BoolTy^\bullet(\false).
% \end{alignat*}
The displayed eliminator $\elimBool^\bullet$ is defined using the elimination principle of $\BoolTy^\bullet$ and the displayed $\beta$-laws hold.
This concludes the definition of all components of $\MS^\bullet$.

\begin{theorem}
  The initial model $\MS$ satisfies canonicity: any closed boolean term $b : \CS.\Tm(1_\CS,\BoolTy)$ is canonical, \ie equal to exactly one of $\true$ or $\false$. 
\end{theorem}
\begin{proof}
  By the relative induction principle~\cref{thm:relative_induction_principle_terminal}, the displayed higher-order model $\MS^{\bullet}$ admits a section $\sem{-}$.
  Now given a closed boolean term $b : \CS.\Tm(1_\CS,\BoolTy)$, we have $\sem{b} : \Tm^{\bullet}(\sem{\BoolTy}, b)$.
  By the computation rule of the section for $\BoolTy$, $\Tm^{\bullet}(\sem{\BoolTy}, b) = \BoolTy^\bullet(b)$.
  Thus, $\sem{b} : \BoolTy^\bullet(b)$ witnesses the fact that $b$ is canonical.

  Since $\sem{\true} = \true^\bullet$, $\sem{\false} = \false^\bullet$ and $\true^\bullet \neq \false^\bullet$, we know that $\true \neq \false$.
\end{proof}

Note that in the canonicity proof, we have not needed to evaluate the section $\sem{-}$ on non-closed types or terms.
Evaluating the section on open types and terms is usually only needed when encountering binders: the evaluation of the section on a closed binder depends on the evaluation of the section on an open type or term.
The following is an example of the computation of the evaluation of the section $\sem{-}$ on the application of the boolean negation function $\lam(\lambda b \mapsto \elimBool(\BoolTy,\false,\true,b))$ to $\true$.
\begin{alignat*}{1}
  & \sem{\app(\lam(\lambda b \mapsto \elimBool(\BoolTy,\false,\true,b)), \true)} \\
  & \quad = \sem{\lam(\lambda b \mapsto \elimBool(\BoolTy,\false,\true,b))}(\sem{\true}) \\
  & \quad = (\lambda b^\bullet \mapsto \sem{\elimBool(\BoolTy,\false,\true,\underline{b})}[\underline{b} \mapsto b^\bullet])(\true^\bullet) \\
  & \quad = \sem{\elimBool(\BoolTy,\false,\true,\underline{b})}[\underline{b} \mapsto \true^\bullet] \\
  & \quad = \elimBool^\bullet(\BoolTy^\bullet, \false^\bullet, \true^\bullet, \true^\bullet) \\
  & \quad = \false^\bullet.
\end{alignat*}

%%% Local Variables:
%%% mode: latex
%%% TeX-master: "main"
%%% End:

%% file: normalization.tex
\section{Example: normalization proof}\label{sec:normalization}

In this section we prove normalization for the initial model $\CS$ of $\Th$ using an induction principle relative to $F : \CRen_\CS \to \CS$, where $\CRen_\CS$ is the category of renamings of $\CS$, \ie the category whose morphisms are the substitutions of $\CS$ that are built out of variables.
We use the alternative definition from~\cite{DBLP:journals/corr/abs-2102-11649} of $\CRen_\CS$ as the initial object in a category of first-order renaming algebras.

\subsection{The category of renamings}

\begin{definition}
  Let $\CC$ be a first-order model of $\Th$.
  A \defemph{higher-order renaming algebra} $\MC$ over $\CC$ consists of:
  \begin{alignat*}{3}
    & \MC.\Var && :{ }
    && \CC.\Ty(1_\CC) \to \UPsh, \\
    & \MC.\var && :{ }
    && (A : \CC.\Ty(1_\CC)) \to \MC.\Var(A) \to \CC.\Tm(1_\CC, A).
       \tag*{\lipicsEnd{}}
  \end{alignat*}
\end{definition}

\begin{definition}
  Let $\CD$ be a first-order model of $\Th$.
  A \defemph{first-order renaming algebra} over $\CD$ is a category $\CC$ with a terminal object along with a functor $F : \CC \to \CD$ that preserves the terminal object and with the structure of a global higher-order renaming algebra $\MC$ over $F^\ast(\Tele_\MD)$ such that $\MC.\Var$ is locally representable and $\MC.\var$ strictly preserves context extensions.
  \lipicsEnd{}
\end{definition}

Equivalently, a first-order renaming algebra over $\CD$ is a CwF $\CC$ together with a CwF morphism $F : \CC \to \CD$ whose action on types is bijective.
(There is a cofibrantly generated factorization system with such morphisms as its right class.
The renaming algebras are the objects in the left class.)
The category of first-order renaming algebras over $\CS$ is locally finitely presentable, and there is an initial first-order renaming algebra $\CRen_\CS$.
The category $\CRen_\CS$ is the category of renamings of $\CS$; in this section we write $F$ for the functor $F : \CRen_\CS \to \CS$.

\subsection{Relative induction principle}

We pose $\CS_{F} \triangleq F^\ast(\CTele_\MS)$; $\CS_{F}$ is an internal first-order model in $\CPsh(\CRen_\CS)$.

\begin{theorem}[Induction principle for $\MS$ relative to $F : \CRen_\CS \to \CS$]\label{thm:relative_induction_principle_renamings} \mbox{}\\
  Let $\MS^{\bullet}$ be a displayed higher-order model over $\CS_{F}$.
  Given the additional data of
  \begin{alignat*}{3}
    & \var^\bullet && :{ }
    && \forall (A : \CS_{F}.\Ty(1_{\CS_{F}}))\ (A^\bullet : \Ty^\bullet(A))\ (x : \Var(A)) \to \Tm^\bullet(A^\bullet, \var(x)),
  \end{alignat*}
  the displayed contextualization $\CScone_{\MS^\bullet}$ admits a section $\sem{-}$ that satisfies the additional computation rule
%  \begin{alignat*}{1}
    $\sem{\var_A(x)} = \var^\bullet(\sem{A}, x).$
%  \end{alignat*}
\end{theorem}
\begin{proof}
  See~\cref{prf:relative_induction_principle_renamings}.
\end{proof}
Note that $\var_A(x)$ is always a closed term of $\CS_{F}$, thus $\sem{\var_A(x)}$ does not depend on any environment.

% \begin{remark}
%   There is a more general relative induction principle over $F : \CRen_\CS \to \CS$, in which the definition of $\var^\bullet$ is already able to use the section $\sem{-}$, \ie{} its specification is approximately
%   \begin{alignat*}{3}
%     & \var^\bullet && :{ }
%     && \forall (A : \CS_{F}.\Ty(1_{\CS_{F}}))\ (x : \Var(A)) \to \Tm^\bullet(\sem{A}, \var(x)).
%   \end{alignat*}
%   However that specification would be circular, since the existence of the section $\sem{-}$ relies on the definition of $\var^\bullet$.
%   Some additional ideas are needed to break this circularity.
%   In practice, the slightly less general relative induction principles, like~\cref{thm:relative_induction_principle_renamings}, seem sufficient for most applications.
%   \lipicsEnd{}
% \end{remark}

\subsection{Normal forms}

Neutrals and normal forms are defined internally to $\CPsh(\CRen_\CS)$, as inductive families
\[ \Ne, \Nf : \forall (A : \CS_{F}.\Ty(1_{\CS_{F}}))\ (a : \CS_{F}.\Tm(1_{\CS_{F}},A)) \to \SSet, \]
generated by the following constructors:
\begin{alignat*}{3}
  % Ne
  & \var^\ne && :{ }
  && (x : \Var(A)) \to \Ne_A(\var(x)), \\
  & \app^\ne && :{ }
  && \Ne_{\Pi(A,B)}(f) \to \Nf_A(a) \to \Ne_{B[a]}(\app(f,a)), \\
  % & \elimBool^\ne && :{ }
  % && \Nf_{P[\true]}(t) \to \Nf_{P[\false]}(f) \to \Ne_{\BoolTy}(b) \to \Ne_{P[b]}(\elimBool(P,t,f,b)), \\
  % Nf
  % & \ne_{\BoolTy}^{\nf} && :{ }
  % && \Ne_{\BoolTy}(b) \to \Nf_{\BoolTy}(b), \\
  % & \true^{\nf} && :{ }
  % && \Nf_{\BoolTy}(\true), \\
  % & \false^{\nf} && :{ }
  % && \Nf_{\BoolTy}(\false), \\
  & \lam^{\nf} && :{ }
  && ((a : \Var(A)) \to \Nf_{B[a]}(b[a])) \to \Nf_{\Pi(A,B)}(\lam(b)).
\end{alignat*}

The goal of normalization is to prove that every term has a unique normal form:
\[ \forall (A : \CS_{F}.\Ty(1_{\CS_{F}}))\ (a : \CS_{F}.\Tm(1_{\CS_{F}},A)) \to \isContr(\Nf_A(a)). \]

This is accomplished in two steps.
First a normalization function is obtained from the relative induction principle, witnessing the existence of normal forms.
Then the uniqueness of normal forms is derived from the stability of the normalization; a fact that is proven by mutual induction on neutrals and normal forms.

\subsection{Normalization displayed model}

We now construct the normalization displayed higher-order model $\MS^{\bullet}$ over $\CS_{F}$.

A displayed type $A^\bullet : \Ty^\bullet(A)$ over a type $A : \CS_{F}.\Ty(1_{\CS_{F}})$ is a triple $(A^\bullet_p, A^\bullet_u, A^\bullet_q)$ consisting of a logical \emph{predicate}
$A^\bullet_p : \CS_{F}.\Tm(1_{\CS_{F}},A) \to \SSet$,
  over the terms of type $A$, valued in the universe of sets;
an \emph{unquoting} (or reflection) function
$A^\bullet_u : (a : \CS_{F}.\Tm(1_{\CS_{F}},A)) \to \Ne_A(a) \to A^\bullet_p(a)$,
  witnessing the fact that any neutral term satisfies the logical predicate $A^\bullet_p$;
a \emph{quoting} (or reification) function
$A^\bullet_q : (a : \CS_{F}.\Tm(1_{\CS_{F}},A)) \to A^\bullet_p(a) \to \Nf_A(a)$,
  witnessing the fact a term satisfying the logical predicate $A^\bullet_p$ admits a normal form.
A displayed term $a^\bullet$ of type $A^\bullet$ over a term $a : \CS_{F}.\Tm(1_{\CS_{F}},A)$ is an element of $A^\bullet_p(a)$:
$\Tm^\bullet(A^\bullet,a) \triangleq A^\bullet_p(a)$.
The logical predicate for $\Pi$-types is defined in the same way as in the canonicity model.
%Note that we write $\Pi_p^\bullet(A^\bullet,B^\bullet)$ instead of the less readable ${(\Pi^\bullet(A^\bullet,B^\bullet))}_p$.
\[ \Pi_p^\bullet(A^\bullet,B^\bullet)(f) \triangleq (\forall a\ a^\bullet \to B^\bullet_p(a^\bullet,\app(f,a))). \]
The unquoting function relies on the unquoting function of the codomain and the quoting function of the domain.
\[ \Pi_u^\bullet(A^\bullet,B^\bullet)(f^\ne) \triangleq \lambda a^\bullet \mapsto B^\bullet_u(a^\bullet, \app^\ne(f^\ne, A^\bullet_q(a^\bullet))). \]
The quoting function says that any element of a $\Pi$-type is a lambda, as implied by the $\eta$-rule.
It relies on the quoting function of the codomain and the unquoting function of the domain, and on the fact that every variable is neutral.
\[ \Pi_q^\bullet(A^\bullet,B^\bullet)(f^\bullet) \triangleq \lam^\nf(\lambda a \mapsto \mathsf{let}\ a^\bullet = A^\bullet_u(\var^\ne(a))\ \mathsf{in}\ B^\bullet_q(a^\bullet, f^\bullet(a^\bullet))). \]
This completes the definition of the displayed higher-order model $\MS^\bullet$.
It remains to check the last hypothesis of the relative induction principle:
\begin{alignat*}{3}
  & \var^\bullet && :{ }
  && \forall A\ A^\bullet\ (x : \Var(A)) \to \Tm^\bullet(A^\bullet, \var(x)), \\
  & \var^\bullet(A^\bullet,x) && \triangleq{ }
  && A^\bullet_u(\var^\ne(x)).
\end{alignat*}

By the relative induction principle~(\cref{thm:relative_induction_principle_renamings}), we obtain a section $\sem{-}$ of $\CScone_{\MS^\bullet}$.
We can then define the normalization function as follows:
\begin{alignat*}{3}
  & \norm && :{ } && \forall A\ (a : \CS_{F}.\Tm(1_{\CS_{F}},A)) \to \Nf_A(a), \\
  & \norm_A(a) && \triangleq{ } && \sem{A}_{q}(\sem{a}).
\end{alignat*}

\vspace{-0.6em}

\subsection{Stability of normalization and uniqueness of normal forms}

Finally, we show the uniqueness of normal forms following \cite{DBLP:phd/ethos/Kaposi17}: we prove that normalization is stable, that is every normal form for a term $a$ is equal to the normal form of $a$ obtained from the normalization function.
As the proof relies on most of the computation rules of the section $\sem{-}$, it is a good example of computations with a section.

\begin{lemma}[Stability]
  Given any normal form $a^\nf : \Nf_A(a)$, we have $a^\nf = \norm_A(a)$.
\end{lemma}
\begin{proof}
  We prove the following two facts, by mutual induction on neutrals and normal forms: \\
  \begin{minipage}{0.5\textwidth}
  \[
  (a^\ne : \Ne_A(a)) \to \sem{a} = \sem{A}_u(a^\ne),
  \]
  \end{minipage}
  \begin{minipage}{0.5\textwidth}
  \[
  (a^\nf : \Nf_A(a)) \to a^\nf = \sem{A}_{q}(\sem{a}).
  \]
  \end{minipage} \\[1em]
  Each case involves some of the computation rules of $\sem{-}$.
  \begin{description}
  \item[Case $a^\ne = \var^\ne(A,x)$]
    \begin{alignat*}{1}
      & \sem{\var_A(x)} \\
      & \quad= \var^\bullet(\sem{A},x)
        \tag{by the computation rule for $\sem{\var(-)}$} \\
      & \quad= \sem{A}_u(a^\ne).
        \tag{by definition of $\var^\bullet$}
    \end{alignat*}

  \item[Case $a^\ne = \app^\ne(f^\ne,a^\nf)$]
    \begin{alignat*}{1}
      & \sem{\app(f,a)} \\
      & \quad= \app^\bullet(\sem{f},\sem{a})
      \tag{by the computation rule for $\sem{\app(-)}$} \\
      & \quad= \sem{f}(\sem{a})
        \tag{by definition of $\app^\bullet$} \\
      & \quad= \sem{\Pi(A,B)}_u(f^\ne,\sem{a})
        \tag{by the induction hypothesis for $f^\ne$} \\
      & \quad= \sem{B[a]}_u(\app^\ne(f^\ne,a^\nf)) .
        \tag{by definition of $\Pi^\bullet_u$ and the induction hypothesis for $a^\nf$}
    \end{alignat*}

  \item[Case $a^\nf = \lam^\nf(b^\nf)$]
    \begin{alignat*}{1}
      & \sem{\Pi(A,B)}_{q}(\sem{\lam(b)}) \\
      & \quad = \Pi_q^\bullet(\sem{A}, \lambda a^\bullet \mapsto \sem{B(\underline{a})}[\underline{a} \mapsto a^\bullet])(\lambda a^\bullet \mapsto \sem{b(\underline{a})}[\underline{a} \mapsto a^\bullet])
        \tag{by the computation rules for $\sem{\Pi(-)}$ and $\sem{\lam}$} \\
      & \quad = \lam^\nf(\lambda a \mapsto \mathsf{let}\ a^\bullet = \sem{A}_u(\var^\ne(a))\ \mathsf{in}\ {(\sem{B(\underline{a})}[\underline{a} \mapsto a^\bullet])}_q(\sem{b(\underline{a})}[\underline{a} \mapsto a^\bullet])
        \tag{by definition of $\Pi^\bullet_q$} \\
      & \quad = \lam^\nf(\lambda a \mapsto {(\sem{B(\underline{a})}[\underline{a} \mapsto \sem{\var(a)}])}_q(\sem{b(\underline{a})}[\underline{a} \mapsto \sem{\var(a)}]))
        \tag{by the computation rule for $\sem{\var(a)}$} \\
      & \quad = \lam^\nf(\lambda a \mapsto \sem{B[\var(a)]}_q(\sem{b[\var(a)]}))
        \tag{by the naturality of $\sem{-}$} \\
      & \quad = \lam^\nf(b^\nf).
        \tag*{(by the induction hypothesis for $b^\nf$) \qedhere{}}
    \end{alignat*}
  \end{description}
\end{proof}

%%% Local Variables:
%%% mode: latex
%%% TeX-master: "main"
%%% End:

%% file: generic_section.tex
\section{The category of sections of a displayed first-order model}\label{sec:generic_section}

The last tool that is needed for the proofs of relative induction principles is the \emph{category of sections} of an internal displayed higher-order model.
This replaces the use of a displayed inserter in our previous work~\cite{DBLP:journals/corr/abs-2102-11649}.

Let $\CM^\bullet$ be a global internal displayed first-order model over a first-order model $\CM$, internally to some presheaf category $\CPsh(\CC)$.
We see $\CM$ as a functor $\CM : \CC \to \CMod_\Th^\op$, as justified by~\cref{prop:eqv_definitions_fom}.
Write $\mathbf{DispMod}_\Th$ for the external category of a first-order model of $\Th$ with a displayed first-order model over it.
There is a forgetful functor $U : \mathbf{DispMod}_\Th \to \CMod_\Th$.
Since the notion of displayed first-order model is also essentially algebraic, we can also view $\CM^\bullet$ as a functor $\CC \to \mathbf{DispMod}_\Th^\op$ such that $U \circ \CM^\bullet = \CM$.
Similarly, writing $\CSect_\Th$ for the category of a displayed first-order model of $\Th$ with a section, a section of $\CM^\bullet$ can be identified with a functor $\sem{-} : \CC \to \CSect^\op_{\Th}$ such that $V \circ \sem{-} = \CM^\bullet$, where $V$ is the forgetful functor $\CSect_\Th \to \mathbf{DispMod}_\Th$.

\begin{definition}
  The category $\CSect_{\Th}^{\op}[\CM^\bullet]$ of \defemph{sections} of $\CM^\bullet$ is the pullback (in $\CCat$)
  \[ \begin{tikzcd}
      \CSect_{\Th}^{\op}[\CM^\bullet]
      \ar[rd, phantom, very near start, "\lrcorner"]
      \ar[d, "\pi_{0}"]
      \ar[r, "{\sem{-}_0}"]
      & \CSect_\Th^\op
      \ar[d, "V"]
      \\
      \CC
      \ar[r, "{\CM^\bullet}"]
      & \mathbf{DispMod}_\Th^\op
    \end{tikzcd}
  \]
\end{definition}

By the universal property of the pullback, the data of a section of $\CM^\bullet$ is equivalent to the data of a section of $\pi_0$ in $\CCat$.
The category $\CSect_{\Th}^{\op}[\CM^\bullet]$ is itself equipped with a section $\sem{-}_0$ of $\pi_0^\ast(\CM^\bullet)$, which is called the \emph{generic section} of $\CM^\bullet$.

In order to prove an induction principle such as~\cref{thm:relative_induction_principle_renamings}, we want to use the initiality of $\CC$ in some category to obtain section of $\pi_0$.
For example, when $\CC$ is the category of renamings, it suffices to equip $\CSect_{\Th}^{\op}[\CM^\bullet]$ with the structure of a renaming algebra that is preserved by $\pi_0$.

This typically involves lifting the terminal object and context extensions of $\CC$ to $\CSect_{\Th}^{\op}[\CM^\bullet]$.
Conditions for the lifting of these finite limits are given in~\cref{ssec:proofs_generic_section}; the initiality of the syntax is needed to lift the terminal object.

%%% Local Variables:
%%% mode: latex
%%% TeX-master: "main"
%%% End:

%% file: proofs.tex
\section{Technical results}\label{sec:omitted}

\subsection{Local representability}

We recall the definition of the notion of locally representable dependent presheaf, which encodes the notion of context extension.
\begin{definition}\label{def:local_rep}
  Let $\CC$ be a category, $X$ be a presheaf over $\CC$ and $Y$ be a dependent presheaf over $X$.
  Then $Y$ is said to be \defemph{locally representable} if for every element $x : X(\Gamma)$, the presheaf
  \begin{alignat*}{1}
    & Y_{\mid x} : {(\CC/\Gamma)}^\op \to \CSet, \\
    & Y_{\mid x}(\Delta, \rho) \triangleq Y(\Delta, x[\rho])
  \end{alignat*}
  is representable.
  The representing object, consisting of an extended context and a projection map, is written $(\Gamma.Y[x], \bm{p}_x)$ and the generic element is written $\bm{q}_x : Y(\Gamma.Y[x], x[\bm{p}_x])$.

  Given any object $\Delta \in \CC$, map $\rho : \Delta \to \Gamma$ and element $y : Y(\Delta, x[\rho])$, we write $\angles{\rho,a}$ for the unique morphism such that $\bm{p}_x \circ \angles{\rho,y} = \rho$ and $\bm{q}_x[\angles{\rho,y}] = y$.
  \lipicsEnd{}
\end{definition}

\subsection{Characterization of the telescopic contextualization}

We define the contextual slices of a first-order model (which are called fibrant slices in~\cite{HomotopyTheoryTypeTheories}).
\begin{definition}[Contextual slice]\label{def:contextual_slice}
  Let $\CC$ be a first-order model of $\Th$.
  Given $\Gamma \in \CC$, the contextual slice $(\CC \sslash \Gamma)$ is the contextual first-order model given by:
  \begin{itemize}
  \item Objects of $(\CC \sslash \Gamma)$ are telescopes (iterated context extensions) over $\Gamma$.
  \item Morphisms from $\Delta_1$ to $\Delta_2$ are morphisms from $\Gamma.\Delta_1$ to $\Gamma.\Delta_2$ in $(\CC / \Gamma)$.
  \item The rest of the structure is inherited from $\CC$ along the projection
    \[ (\CC \sslash \Gamma) \ni \Delta \mapsto \Gamma.\Delta \in \CC.
      \tag*{\lipicsEnd{}}
    \]
  \end{itemize}
\end{definition}

The contextual slice is functorial in both $\CC$ and (contravariantly) $\Gamma$:
\begin{itemize}
\item 
  For any $f : \Delta \to \Gamma$, there is a pullback morphism $f^\ast : (\CC \sslash \Gamma) \to (\CC \sslash \Delta)$.
  Its actions on objects, substitutions, types and terms are all given by substitution along $f$.
\item
  For any $F : \CC \to \CD$, there is a morphism $(F \sslash \Gamma) : (\CC \sslash \Gamma) \to (\CD \sslash F(\Gamma))$.
  Its actions on objects, substitutions, types and terms are given by the actions of $F$ on telescopes, substitutions, types and terms.
\item
  The following diagrams commute (for any $F : \CC \to \CD$ and $f : \Delta \to \Gamma$):
  \[ \begin{tikzcd}
      (\CC \sslash \Gamma)
      \ar[r, "(F \sslash \Gamma)"]
      \ar[d, "f^\ast"]
      & (\CD \sslash F(\Gamma))
      \ar[d, "f^\ast"]
      \\
      (\CC \sslash \Delta)
      \ar[r, "(F \sslash \Delta)"]
      & (\CD \sslash F(\Delta)) \rlap{\ .}
    \end{tikzcd} \]
\item For any $f : \Theta \to \Gamma$ and object $\Delta$ of $(\CC \sslash \Gamma)$, we have
  $(f^\ast \sslash \Delta) = \angles{\bm{p}_\Delta^\ast(f),\bm{q}_\Delta}^\ast$ as a morphism from $(\CC \sslash \Gamma.\Delta)$ to $(\CC \sslash \Theta.f^\ast(\Delta))$.
  Here $\angles{f \circ \bm{p}_{f^\ast(\Delta)},\bm{q}_{f^\ast(\Delta)}}$ is a morphism from $\Theta.f^\ast(\Delta)$ to $\Gamma.\Delta$.
\end{itemize}

\begin{lemma}\label{lem:contextual_slice}
  Let $\CC$ be a first-order model of $\Th$.
  Then the telescopic contextualization $\CTele_\MC$ corresponds the contextual slice functor
  \[ (\CC \sslash -) : \CC \to \CMod_\Th^\op.
    \tag*{\lipicsEnd{}}
  \]
\end{lemma}
\begin{proof}
  Immediate from unfolding the definitions.
\end{proof}

\begin{corollary}\label{lem:contextual_slice_terminal}
  The externalization $1_\CC^\ast(\CTele_\MC)$ is the contextual core of $\CC$.
\end{corollary}
\begin{proof}
  By~\cref{lem:contextual_slice}, $1_\CC^\ast(\CTele_\MC)$ is the contextual slice $(\CC \sslash 1_\CC)$, the contextual core of $\CC$.
\end{proof}

\begin{lemma}\label{lem:contextual_slice_ext_closed}
  Let $\CC$ be a contextual first-order model and $A : \CC.\Ty(1_\CC)$ be a closed type.
  Then the contextual slice $(\CC \sslash A)$ satisfies the following universal property:
  for every model $\CE$, morphism $F : \CC \to \CE$ and element $a : \CE.\Tm(F(A))$, there is a unique morphism $\widetilde{F} : (\CC \sslash 1_\CC.A) \to \CE$ such that $\widetilde{F} \circ \bm{p}_A^{\ast} = F$ and $\widetilde{F}(\bm{q}_A) = a$.

  In other words, $(\CC \sslash 1_\CC.A)$ is the free extension of $\CC$ by a generic element $\bm{q}_A$ of type $A$.
\end{lemma}
\begin{proof}
  The contextuality of $\CC$ implies that $\CC \cong (\CC \sslash 1_\CC)$ and thus we have a weakening morphism $\bm{p}_A^\ast : \CC \to (\CC \sslash 1_\CC.A)$.
  We write $A' : \CC.\Ty(1_\CC.A)$ for the weakening of $A$, \ie $A' = \bm{p}_A^\ast(A)$.
  
  Take a morphism $F : \CC \to \CE$ and an element $a : \CE.\Tm(F(A))$.

  Let $\tilde{F} : (\CC \sslash 1_\CC.A) \to \CE$ be any morphism such that $\tilde{F} \circ \bm{p}_A^{\ast} = F$ and $\tilde{F}(\bm{q}_A) = a$.
  Since $(\CC \sslash 1_\CC.A)$ is contextual, $\tilde{F}$ factors through $(\tilde{F} \sslash (1_\CC.A)) : (\CC \sslash 1_\CC.A) \to (\CE \sslash 1_\CE)$.

  The following diagram commutes:
  \[ \begin{tikzcd}[column sep=50pt]
      (\CC \sslash 1_\CC.A)
      \ar[d, "\bm{p}_{A'}^{\ast}"]
      \ar[r, "(\tilde{F} \sslash (1_\CC.A))"]
      \ar[dd, equal, bend right=90]
      &
      (\CE \sslash 1_\CE)
      \ar[d, "{\bm{p}_{\tilde{F}(A')}^{\ast}}"]
      \ar[dd, equal, bend left=90]
      \\
      (\CC \sslash 1_\CC.A.A')
      \ar[r, "{(\tilde{F} \sslash (1_\CC.A).A')}"]
      \ar[d, "\angles{\bm{q}_{A}}^{\ast}"]
      &
      (\CE \sslash 1_\CE.\tilde{F}(A'))
      \ar[d, "\angles{\tilde{F}(\bm{q}_A)}^\ast"]
      \\
      (\CC \sslash 1_\CC.A)
      \ar[r, "(\tilde{F} \sslash (1_\CC.A))"]
      &
      (\CE \sslash 1_\CE) \rlap{\ .}
    \end{tikzcd} \]

  By assumptions, we have $\tilde{F}(A') = F(A)$ and $\tilde{F}(\bm{q}_A) = a$.
  Thus, $\tilde{F}$ factors as the composition
  \[ (\CC \sslash 1_\CC.A) \xrightarrow{\bm{p}_{A'}^\ast} (\CC \sslash 1_\CC.A.A') \xrightarrow{(\tilde{F} \sslash (1_\CC.A).A')} (\CE \sslash 1_\CE.F(A)) \xrightarrow{\angles{a}^\ast} (\CE \sslash 1_\CE) \xrightarrow{} \CE. \]

  By functoriality of $(- \sslash 1_\CC.A)$, we have $(F \sslash 1_\CC.A) = (\tilde{F} \sslash (1_\CC.A).A') \circ (\bm{p}_A^\ast \sslash 1_\CC.A)$.
  We can simplify $(\bm{p}_A^\ast \sslash 1_\CC.A) = \angles{\bm{p}_A \circ \bm{p}_{A'}, \bm{q}_{A'}}^\ast$.
  Finally, we see that $\angles{\bm{q}_A}^\ast \circ \angles{\bm{p}_A \circ \bm{p}_{A'}, \bm{q}_{A'}}^\ast = \id = \angles{\bm{q}_A}^\ast \circ \bm{p}_{A'}^{\ast}$.
 
  Thus, $(\tilde{F} \sslash (1_\CC.A).A') \circ \bm{p}_{A'}^\ast = F_\ast(1_\CC.A)$ and we can simplify the factorization to
  \[ (\CC \sslash 1_\CC.A) \xrightarrow{F_\ast(1_\CC.A)} (\CE \sslash 1_\CE.F(A)) \xrightarrow{\angles{a}^\ast} (\CE \sslash 1_\CE) \xrightarrow{} \CE. \]
  
  As this composition does not depend on $\tilde{F}$ and satisfies the required equations, this shows the existence and uniqueness of $\tilde{F}$.
 \end{proof}

\begin{lemma}\label{lem:contextual_slice_ext}
  Given any $\Gamma \in \CC$ and $A : \CC.\Ty(\Gamma)$, then $(\CC \sslash \Gamma.A)$ satisfies the following universal property:
  for every model $\CE$, morphism $F : (\CC \sslash \Gamma) \to \CE$ and element $a : \CE.\Tm(F(A))$, there is a unique morphism $\widetilde{F} : (\CC \sslash \Gamma.A) \to \CE$ such that $\widetilde{F} \circ \bm{p}_A^\ast = F$ and $\widetilde{F}(\bm{q}_A) = a$.

  In other words, $(\CC \sslash \Gamma.A)$ is the free extension of $(\CC \sslash \Gamma)$ by a generic element $\bm{q}_A$ of type $A$.
\end{lemma}
\begin{proof}
  We have $(\CC \sslash \Gamma.A) \cong ((\CC \sslash \Gamma) \sslash 1_{(\CC\sslash\Gamma)}.A)$.
  Then the result follows from~\cref{lem:contextual_slice_ext_closed}.
\end{proof}

\subsection{Properties of the category of sections}\label{ssec:proofs_generic_section}

We now prove the properties of the category of sections of a displayed first-order model.
Fix a global internal first-order model $\CM : \CC \to \CMod_\Th^\op$ and a displayed first-order model $\CM^\bullet : \CC \to \mathbf{DispMod}_\Th^\op$ over $\CM$.

We prove conditions that relate the existence of some finite limits in $\CSect_{\Th}^{\op}[\CM^\bullet]$ to the universal properties of the first-order models $\CM(\Gamma)$ for $\Gamma \in \CC$.

We have defined the category $\CSect_{\Th}^{\op}[\CM^\bullet]$ of sections of $\CM^\bullet$ as the following pullback
\[ \begin{tikzcd}
    \CSect_{\Th}^{\op}[\CM^\bullet]
    \ar[rd, phantom, very near start, "\lrcorner"]
    \ar[d, "\pi_{0}"]
    \ar[r, "{\sem{-}_0}"]
    & \CSect_\Th^\op
    \ar[d]
    \\
    \CC
    \ar[r, "{\CM^\bullet}"]
    & \mathbf{DispModel}_\Th^\op \rlap{\ .}
  \end{tikzcd}
\]

Unfolding the definition, an object of $\CSect_{\Th}^{\op}[\CM^\bullet]$ is a pair $(\Gamma,\sem{-}_\Gamma)$ where $\Gamma \in \CC$ and $\sem{-}_\Gamma$ is a section of $\CM^\bullet(\Gamma)$ over $\CM(\Gamma)$.
A morphism of from $(\Gamma,\sem{-}_\Gamma)$ to $(\Delta,\sem{-}_\Delta)$ is a morphism $f : \CC(\Gamma,\Delta)$ such that the outer square commutes in the following diagram:
\[ \begin{tikzcd}
    \CM^\bullet(\Delta)
    \ar[r, "{\CM^\bullet(f)}"]
    \ar[d, -{Triangle[open]}]
    & \CM^\bullet(\Gamma)
    \ar[d, -{Triangle[open]}]
    \\ \CM(\Delta)
    \ar[r, "{\CM(f)}"]
    \ar[u, "{\sem{-}_\Gamma}", bend left]
    & \CM(\Gamma) \rlap{\ .}
    \ar[u, "{\sem{-}_\Delta}"', bend right]
  \end{tikzcd} \]

\begin{lemma}\label{lem:generic_section_terminal}
  If $\CC$ has a terminal object $1_\CC$ and $\CM(1_\CC)$ is the initial model of $\Th$, then $\CSect^{\op}_{\CM^\bullet}$ has a terminal object that is strictly preserved by $\pi_0$.
\end{lemma}
\begin{proof}
  Since $\CM(1_\CC)$ is initial, we obtain a section $\sem{-}_{1_\CC}$ of $\CM^\bullet(1_\CC)$.
  We now prove that $(1_\CC,\sem{-}_{1_\CC})$ is terminal in $\CSect^{\op}_{\CM^\bullet}$.
  Let $(\Gamma,\sem{-}_\Gamma)$ be any object of $\CSect^{\op}_{\CM^\bullet}$.
  A morphism from $(\Gamma,\sem{-}_\Gamma)$ to $(1_\CC,\sem{-}_{1_\CC})$ is a morphism $f : \CC(\Gamma,1_\CC)$ such that the following square commutes:
  \[ \begin{tikzcd}
      \CM^\bullet(1_\CC)
      \ar[r, "\CM^\bullet(f)"]
      & \CM^\bullet(\Gamma)
      \\
      \CM(1_\CC)
      \ar[u, "{\sem{-}_{1_\CC}}"]
      \ar[r, "\CM(f)"]
      & \CM(\Gamma) \rlap{\ .}
      \ar[u, "{\sem{-}_{\Gamma}}"']
    \end{tikzcd} \]

  Since $1_\CC$ is terminal, there is only one morphism $f : \CC(\Gamma,1_\CC)$, and the corresponding square commutes by initiality of $\CM(1_\CC)$.
\end{proof}

\begin{definition}\label{def:internal_compatibility_condition}
  Let $X$ be a presheaf over $\CC$ and $Y$ be a locally representable dependent presheaf over $X$.
  Assume given the data of global elements
  \begin{alignat*}{3}
    & f_X && :{ }
    && X \to \CM.\Ty(1_\CM), \\
    & f_Y && :{ }
    && (x : X) \to Y(x) \to \CM.\Tm(1_\CM, f_X(x))
  \end{alignat*}
  of $\CPsh(\CC)$.

  We say that $f_Y$ is compatible with $\CM$ if for every element $x : X(\Gamma)$, the external first-order model $\CM(\Gamma.Y(x))$ satisfies the following universal property:
  for every first-order model $\CE$, morphism $E : \CM(\Gamma) \to \CE$ and element $z : \CE.\Tm(1_\CE, E(f_X(x)))$, there is a unique morphism $\widetilde{E} : \CM(\Gamma.Y(x)) \to \CE$ such that $E = \widetilde{E} \circ \CM(\bm{p}_x)$ and $z = \widetilde{E}(f_Y(x, \bm{q}_x))$.

  In other words, $\CM(\Gamma.Y(x))$ should be the free extension of $\CM(\Gamma)$ by an element $f_Y(x, \bm{q}_x)$ of type $f_X(x)$, the extension being witnessed by the morphism $\CM(\bm{p}_x) : \CM(\Gamma) \to \CM(\Gamma.Y(x))$.
  \lipicsEnd{}
\end{definition}

\begin{lemma}\label{lem:cwf_morphism_compatibility}
  Let $F : \CC \to \CD$ be a CwF morphism.
  
  Then the condition of~\cref{def:internal_compatibility_condition} is satisfied, with $\CM = F^\ast(\Tele_\MD)$, $X$ and $Y$ being respectively the types and terms of $\CC$, and $f_X$ and $f_Y$ being the actions of $F$ on types and terms.
\end{lemma}
\begin{proof}
  By~\cref{def:contextual_slice}, $\CM$ correspond to the functor
  \[ \CC \ni \Gamma \mapsto (\CD \sslash F(\Gamma)) \in \CMod_\Th^\op. \]

  Thus, we need to check that given any $\Gamma \in \CC$ and $A : \CC.\Ty(\Gamma)$, the model $(\CD \sslash F(\Gamma.A))$ is the free extension of $(\CD \sslash F(\Gamma))$ by a generic element of type $F(A)$.
  Since $F$ preserves extensions, $F(\Gamma.A) \cong F(\Gamma).F(A)$, and thus the result follows by~\cref{lem:contextual_slice_ext}.
\end{proof}

\begin{lemma}\label{lem:generic_section_locally_representable}
  Let $X$ be a presheaf over $\CC$ and $Y$ be a locally representable presheaf family over $X$ equipped with operations $f_X$ and $f_Y$ satisfying the condition of~\cref{def:internal_compatibility_condition}.
  Finally, assume that for every $(\Gamma,\sem{-}_\Gamma) \in \CSect^\op_{\CM^\bullet}$, $\Delta \in \CC$, $\gamma : \CC(\Delta,\Gamma)$, $x : X(\Gamma)$ and $y : Y(\Delta,x[\gamma])$ we have
  \begin{alignat*}{3}
    & f_Y^\bullet(x,y) && :{ }
    && \CM^\bullet(\Delta).\Tm^\bullet(1^\bullet, \CM^\bullet(\gamma)(\sem{f_X(x)}_\Gamma), f_Y(x[\gamma],y)),
  \end{alignat*}
  naturally in $(\Gamma,\sem{-}_\Gamma)$ and $\Delta$.

  Consider the presheaf $X_0 \triangleq \pi_0^\ast(X)$ over $\CSect^{\op}_{\CM^\bullet}$ and the dependent presheaf $Y_0$ over $X_0$ specified on objects by:
  \begin{alignat*}{3}
    & Y_0((\Gamma,\sem{-}_\Gamma),x) && \triangleq{ }
    && \{ y : Y(\Gamma,x) \mid \sem{f_Y(x,y)}_\Gamma = f_Y^\bullet(x,y) \}.
  \end{alignat*}
  Then the presheaf family $Y_0$ is locally representable and the action induced by the first projections $Y_0((\Gamma,\sem{-}_\Gamma),x) \to Y(\Gamma,x)$ strictly preserve context extensions.
\end{lemma}
\begin{proof}
  Let $(\Gamma,\sem{-}_\Gamma)$ be an object of $\CSect^{\op}_{\CM^\bullet}$ and $x : X(\Gamma)$ be an element of $X_0$ at this object.
  We have to prove that the presheaf ${Y_0}_{\mid x}$ over $(\CSect^{\op}_{\CM^\bullet} / (\Gamma,\sem{-}_\Gamma))$ is representable.

  Consider the diagram:
  \[ \begin{tikzcd}
      \CM^\bullet(\Gamma)
      \ar[r, "\CM^\bullet(\bm{p}_x)"]
      \ar[d, -{Triangle[open]}]
      & \CM^\bullet(\Gamma.Y(x))
      \ar[d, -{Triangle[open]}]
      \\ \CM(\Gamma)
      \ar[r, "\CM(\bm{p}_x)"]
      \ar[u, bend left, "\sem{-}_{\Gamma}"]
      & \CM(\Gamma.Y(x))
      \ar[u, bend right, dashed, "\sem{-}_{\Gamma.Y(x)}"']
    \end{tikzcd} \]

  We construct a section $\sem{-}_{\Gamma.Y(x)}$ of $\CM^\bullet(\Gamma.Y(x))$ over $\CM(\Gamma.Y(x))$.
  Using the universal property of $\CM(\Gamma.Y(x))$, we define $\sem{-}_{\Gamma.Y(x)}$ as the unique extension of $\CM^\bullet(\bm{p}_x) \circ \sem{-}_\Gamma$ that sends $f_Y(x,\bm{q}_x)$ to $f^\bullet_Y(x,\bm{q}_x)$.

  We can check that $\bm{p}_x$ lifts to a morphism $\bm{p}_x : (\Gamma.Y(x), \sem{-}_{\Gamma.Y(x)}) \to (\Gamma, \sem{-}_\Gamma)$ in $\CSect^{\op}_{\CM^\bullet}$.

  We now show that $((\Gamma.Y(x), \sem{-}_{\Gamma.Y(x)}), \bm{p}_x)$ represents the functor ${Y_0}_{\mid x}$.
  Let $(\Delta, \sem{-}_\Delta)$ be another object of $\CSect^{\op}_{\CM^\bullet}$, with a morphism $\rho : (\Delta, \sem{-}_\Delta) \to (\Gamma, \sem{-}_\Gamma)$ and an element $y : {Y_0}_{\mid x}((\Delta, \sem{-}_\Delta), \rho)$.
  Unfolding the definitions, we have $y : Y(\Delta, x[\rho])$ with $\sem{f_Y(x[\rho], y)}_\Delta = f_Y^\bullet(x[\rho], y)$.

  The local representability of $Y$ implies that there is a unique morphism $\widetilde{\rho} : \Delta \to \Gamma.Y(x)$ in $\CC$ such that $\bm{p}_x \circ \widetilde{\rho} = \rho$ and $\bm{q}_x[\widetilde{\rho}] = y$.
  We have to show that this morphism lifts to $\CSect^{\op}_{\CM^\bullet}$, \ie{} that the following square commutes:
  \[ \begin{tikzcd}
      \CM^\bullet(\Gamma.Y(x))
      \ar[r, "\CM^\bullet(\widetilde{\rho})"]
      \ar[d, -{Triangle[open]}]
      & \CM^\bullet(\Delta)
      \ar[d, -{Triangle[open]}]
      \\ \CM(\Gamma.Y(x))
      \ar[r, "\CM(\widetilde{\rho})"]
      \ar[u, bend left, "\sem{-}_{\Gamma.Y(x)}"]
      & \CM(\Delta)
      \ar[u, bend right, "\sem{-}_{\Delta}"']
    \end{tikzcd} \]
  By the universal property of $\CM(\Gamma.Y(x))$, it suffices to show that $f_Y(x,\bm{q}_x)$ is mapped to the same element by the compositions $\CM^\bullet(\widetilde{\rho}) \circ \sem{-}_{\Gamma.Y(x)}$ and $\sem{-}_\Delta \circ \CM(\widetilde{\rho})$.
  We compute $\CM^{\bullet}(\widetilde{\rho})(\sem{f_Y(x,\bm{q}_x)}_{\Gamma.Y(x)}) = \CM^{\bullet}(\widetilde{\rho})(f_Y^\bullet(x, \bm{q}_x)) = f_Y^\bullet(x[\rho], y)$ and $\sem{\CM(\widetilde{\rho})(f_Y(x,\bm{q}_x))}_\Delta = \sem{f_Y(x[\rho],y)}_\Delta = f_Y^\bullet(x[\rho], y)$.

  This completes the proof that $((\Gamma.Y(x), \sem{-}_{\Gamma.Y(x)}), \bm{p}_x)$ represents the functor ${Y_0}_{\mid x}$.

  We have proven that ${Y_0}_{\mid x}$ is representable for every $x$, \ie that $Y_0$ is locally representable.
  The first projections $Y_0((\Gamma,\sem{-}_\Gamma),x) \to Y(\Gamma,x)$ strictly preserve the chosen representing objects.
\end{proof}

\subsection{Proofs of relative induction principles}

\begin{proof}[Proof of~\cref{thm:relative_induction_principle_renamings}]\label{prf:relative_induction_principle_renamings}
  We consider the category $\CSect^{\op}_{\Th}[\CScone_{\MS^\bullet}]$ of sections of $\CScone_{\MS^\bullet}$.

  By~\cref{lem:generic_section_terminal}, the category $\CSect^{\op}_{\Th}[\CScone_{\MS^\bullet}]$ has a terminal object.
  We equip it with the structure of a higher-order renaming algebra $(\Var_0,\var_0)$ over $(F \circ \pi_0) : \CSect^{\op}_{\Th}[\CScone_{\MS^\bullet}] \to \CS$ as follows:
  \begin{alignat*}{1}
    & \Var_0((\Gamma,\sem{-}_\Gamma),A) \triangleq
    \{ a : \Var(\Gamma,A) \mid \sem{ \var(\Gamma,a) }_\Gamma = \var^\bullet(\Gamma,\sem{ A }_\Gamma, a) \}, \\
    & \var_0((\Gamma,\sem{-}_\Gamma),A,a) \triangleq
      \var(\Gamma,a).
  \end{alignat*}

  By~\cref{lem:cwf_morphism_compatibility}, the action of $F : \CR \to \CS$ on variables is compatible with $\CS_F$.
  By~\cref{lem:generic_section_locally_representable}, the presheaf family $\Var_0$ is locally representable and the first projections $\Var_0((\Gamma,\sem{-}_\Gamma),A) \to \Var(\Gamma,A)$ strictly preserve context extensions.
  By initiality of $\CR$ among first-order renaming algebras, we obtain a section $H$ of $\pi_0$ in the category of first-order renaming algebras.

  We thus have a section $\sem{-} \triangleq H^\ast(\sem{-}_0)$ of $\CScone_{\MS^\bullet}$ in $\CPsh(\CR)$.
  The action of $H$ on variables proves that it satisfies the equality $\sem{\var_A(x)} = \var^\bullet(\sem{A}, x)$.
\end{proof}

%%% Local Variables:
%%% mode: latex
%%% TeX-master: "main"
%%% End:

%% file: example_syntactic_parametricity.tex
\section{Example: Syntactic parametricity}\label{sec:parametricity}

In this appendix, we show how our constructions can be combined to obtain a syntactic parametricity translation for a dependent type theory with universes.
Syntactic parametricity \cite{bernardy12parametricity,popl16} refers to the situation where the result of the translation itself consists of types and terms from the syntax, as opposed to metatheoretic entities such as sets or presheaves.

Syntactic parametricity is an interesting application of our methods because the motives and methods of the induction are given over some category (here the syntax), but the result of the induction is only obtained over some other category (here the terminal category, selecting the empty context).
We construct a displayed higher-order model and its $\CScone$-contextualization in presheaves over the syntax, but we only look at the section of its externalization, and never obtain a section of the full $\CScone$-contextualization.

Let $\Th$ be a dependent type theory with a hierarchy of Coquand universes~\cite{CoquandPresheafNote,DBLP:conf/csl/Kovacs22}, such that types at every level are closed under $\Pi$-types.
This means that a higher-order model of $\Th$ consists of the following data, where every line quantifies over a universe level $i \in \Nat$:
\begin{alignat*}{1}
  & \Ty_i : \SSet, \\
  & \Tm_i : \Ty_i \to \SSet, \\
  & \UU_i : \Ty_{i+1}, \\
  & \El_i : \Tm(\UU_i) \cong \Ty_{i}, \\
  & \mathsf{Lift}_i : \Ty_i \to \Ty_{i+1}, \\
  & \mathsf{lift}_i : \Tm_i(A) \cong \Tm_{i+1}(\mathsf{Lift}_i(A)), \\
  & \Pi : (A : \Ty_i) \to (\Tm_i(A) \to \Ty_i) \to \Ty_i, \\
  & \app : \Tm_i(\Pi(A,B)) \cong ((a : \Tm_i(A)) \to \Tm_i(B(a))).
\end{alignat*}
We write $\CS$ for the syntax of $\Th$, that is the initial first-order model of $\Th$; the constructions of the paper generalize to this theory.

We construct, internally to $\CPsh(\CS)$, a displayed higher-order model $\MS^\bullet$ over $\Tele_\MS$:
\begin{itemize}
\item A displayed type over $A : \MS.\Ty_i$ is a type family
  \[ A^\bullet : \MS.\Tm(A) \to \MS.\Ty_{i}. \]
\item A displayed term of type $A^\bullet$ over $a : \MS.\Tm(A)$ is an element
  \[ a^\bullet : \MS.\Tm(A^\bullet(a)). \]
\item The displayed universe $\UU^\bullet_i$ is defined as
  \[ \UU^\bullet_i \triangleq \lambda A \mapsto \Pi(\El(A), \UU_i), \]
  so that we have an isomorphism $\UU^\bullet_i(A) \cong \Ty^\bullet_i(\El_i(A))$.
\item The displayed lifting operation $\mathsf{Lift}^\bullet_i$ is defined by
  \[ \mathsf{Lift}^\bullet_i(A^\bullet) \triangleq \lambda a \mapsto \mathsf{Lift}_i(A^\bullet(\mathsf{lift}_i^{-1}(a))); \]
  we have an isomorphism $\Tm^\bullet_i(A^\bullet,a) \cong \Tm^\bullet_{i+1}(\mathsf{Lift}^\bullet_i(A^\bullet), \mathsf{lift}_i(a))$.
\item The displayed $\Pi$-type $\Pi^\bullet(A^\bullet,B^\bullet)$ is defined as
  \[ \Pi^\bullet(A^\bullet,B^\bullet) \triangleq \lambda f \mapsto (\Pi(A, \lambda a \mapsto \Pi(A^\bullet(a), \lambda a^\bullet\mapsto B^\bullet(a^\bullet, \app(f,a)))).\]
\end{itemize}

By the induction principle of $\CS$ relative to $1_\CS : 1_\CCat \to \CS$, we have a section $\sem{-}$ of $1_\CS^\ast(\CScone_{\MS^\bullet})$.
Let $f$ be any closed term of $\CS$ of type $(A : \UU_0) \to A \to A$.

Then $\sem{f}$ is a closed term of $\CS$ of type $\sem{(A : \UU_0) \to A \to A}(f)$ witnessing the fact that $f$ satisfies parametricity.
Using the computation rules of $\sem{-}$, we can compute
\begin{alignat*}{1}
  & \sem{(A : \UU_0) \to A \to A}(f) \\
  & \quad = (A : \UU_0)(A^\bullet : A \to \UU_0) \to \sem{\lambda A \mapsto A \to A}(A^\bullet,f(A)) \\
  & \quad = (A : \UU_0)(A^\bullet : A \to \UU_0)(a : A)(a^\bullet : A^\bullet(A)) \to \sem{\lambda (A,a) \mapsto A}((A^\bullet,a^\bullet),f(A,a)) \\
  & \quad = (A : \UU_0)(A^\bullet : A \to \UU_0)(a : A)(a^\bullet : A^\bullet(A)) \to A^\bullet(f(A,a)).
\end{alignat*}

%%% Local Variables:
%%% mode: latex
%%% TeX-master: "main"
%%% End:

%% file: example_cubical.tex
\section{Induction principles for cubical type theories}\label{sec:cubical}

In this section we state and prove relative induction principles for a minimal cubical type theory.
It should be possible to use these induction principles (extended to larger cubical type theories) in order to express in our framework the known proofs of strict canonicity, homotopy canonicity~\cite{DBLP:journals/lmcs/CoquandHS22} and normalization~\cite{NormalizationCTT} for cubical type theory.

\begin{definition}
  A \defemph{higher-order cubical algebra} is a set $\MI$ with two points $0, 1 : \MI$.
  \lipicsEnd{}
\end{definition}

\begin{definition}
  A \defemph{first-order cubical algebra} is a category $\CC$ with a terminal object $1_\CC$, along with a locally representable presheaf $\MI$ with two global elements $0, 1$.
  \lipicsEnd{}
\end{definition}

The cube category $\square$ is the initial first-order cubical algebra.

We consider a minimal cubical type theory $\Th_{\mathsf{CTT}}$ with only $\Pi$-types and path types.
A higher-order model of $\Th_{\mathsf{CTT}}$ consists of the following data:
\begin{alignat*}{1}
  & \MI : \SSet, \\
  & 0,1 : \MI, \\
  & \Ty : \SSet, \\
  & \Tm : \Ty \to \SSet, \\
  & \Pi : (A : \Ty) \to (\Tm(A) \to \Ty) \to \Ty, \\
  & \app : \Tm(\Pi(A,B)) \cong ((a : \Tm(A)) \to \Tm(B(a))), \\
  & \mathsf{Path} : (A : \MI \to \Ty) \to \Tm(A(0)) \to \Tm(A(1)) \to \Ty, \\
  & \mathsf{papp} : \Tm(\mathsf{Path}(A,x,y)) \to (i : \MI) \to \Tm(A(i)), \\
  & (p : \Tm(\mathsf{Path}(A,x,y))) \to \mathsf{papp}(p,0) = x, \\
  & (p : \Tm(\mathsf{Path}(A,x,y))) \to \mathsf{papp}(p,1) = y, \\
  & \mathsf{plam} : (A : \MI \to \Ty) \to (p : (i : \MI) \to \Tm(A(i))) \to \Tm(\mathsf{Path}(A,p(0),p(1))).
\end{alignat*}

We write $\CS_{\mathsf{CTT}}$ for the initial first-order model of $\Th_{\mathsf{CTT}}$.

By the universal property of $\square$, the cubical algebra $(\CS_{\mathsf{CTT}}.\MI,\CS_{\mathsf{CTT}}.0,\CS_{\mathsf{CTT}}.1)$ specifies a morphism $F : \square \to \CS_{\mathsf{CTT}}$ of first-order cubical algebras.
We pose $\CS_F \triangleq F^\ast(\CTele_{\MS_{\mathsf{CTT}}})$.
We write $\mathsf{int} : \square.\MI \to \CS_F.\MI$ for the action of $F$ on elements of the interval in $\CPsh(\square)$.

\begin{theorem}[Induction principle relative to $F : \square \to \CS_{\mathrm{CTT}}$]\label{thm:relative_induction_principle_cubical}
  Let $\MS^\bullet$ be a displayed higher-order model of CTT over $\CS_F$, and assume
  given the additional data of a function
  \begin{alignat*}{1}
    & \mathsf{int}^{\bullet} : (i : \square.\MI) \to \MI^\bullet(\mathsf{int}(i))
  \end{alignat*}
  such that
  \begin{alignat*}{1}
    & \mathsf{int}^{\bullet}(0) = 0^\bullet, \\
    & \mathsf{int}^{\bullet}(1) = 1^\bullet.
  \end{alignat*}

  Then there exists a section $\sem{-}$ of $\CScone_{\MS^\bullet}$. It satisfies the additional equality
  \begin{alignat*}{1}
    & \sem{\mathsf{int}(i)} = \mathsf{int}^\bullet(i).
  \end{alignat*}
\end{theorem}
\begin{proof}
  We consider the category $\CSect^{\op}_{\Th}[\CScone_{\MS^\bullet}]$ of sections of $\CScone_{\MS^\bullet}$.

  By~\cref{lem:generic_section_terminal}, the category $\CSect^{\op}_{\Th}[\CScone_{\MS^\bullet}]$ has a terminal object.
  We equip it with a higher-order cubical algebra $(\MI_0,0_0,1_0)$ as follows:
  \begin{alignat*}{1}
    & \MI_0(\Gamma,\sem{-}_\Gamma) \triangleq
      \{ i : \square.\MI(\Gamma) \mid \sem{ \mathsf{int}(i) }_\Gamma = \mathsf{int}^\bullet(\Gamma, i) \}, \\
    & 0_0 \triangleq \square.0, \\
    & 1_0 \triangleq \square.1.
  \end{alignat*}

  The required equalities $\sem{\mathsf{int}(\square.0)}_\Gamma = \mathsf{int}^\bullet(\Gamma,\square.0)$ and $\sem{\mathsf{int}(\square.1)}_\Gamma = \mathsf{int}^\bullet(\Gamma,\square.1)$ follow from the assumptions
  $\mathsf{int}^{\bullet}(\square.0) = 0^\bullet$ and $\mathsf{int}^{\bullet}(\square.1) = 1^\bullet$.

  By~\cref{lem:cwf_morphism_compatibility}, the action of $F : \square \to \CS$ on the elements of $\square.\MI$ is compatible with $\CS_F$.
  By~\cref{lem:generic_section_locally_representable}, the presheaf $\MI_0$ is locally representable and the first projections $\MI_0(\Gamma,\sem{-}_\Gamma) \to \square.\MI(\Gamma)$ strictly preserve context extensions.
  By initiality of $\square$ among first-order cubical algebras, we obtain a section $H$ of $\pi_0$ in the category of first-order cubical algebras.

  We thus have a section $\sem{-} \triangleq H^\ast(\sem{-}_0)$ of $\CScone_{\MS^\bullet}$ in $\CPsh(\square)$.
  The action of $H$ on elements of the interval proves that $\sem{-}$ satisfies the equality $\sem{\mathsf{int}(i)} = \mathsf{int}^\bullet(i)$.
\end{proof}

\begin{definition}
  Let $\CC$ be a first-order model of $\Th_{\mathsf{CTT}}$.
  A \defemph{higher-order cubical renaming algebra} $\MC$ over $\CC$ consists of a higher-order renaming algebra $(\MC.\Var,\MC.\var)$ over $\CC$, a higher-order cubical algebra $(\MC.\MI,\MC.0,\MC.1)$, and an operation
  \[ \MC.\mathsf{int} : \MC.\MI \to \CC.\MI(1_\CC) \]
  such that $\MC.\mathsf{int}(0) = 0$ and $\MC.\mathsf{int}(1) = 1$.
  \lipicsEnd{}
\end{definition}

\begin{definition}
  Let $\CD$ be a first-order model of $\Th_{\mathsf{CTT}}$.
  A \defemph{first-order cubical renaming algebra} over $\CD$ is a category $\CC$ with a terminal object along with a functor $F : \CC \to \CD$ that preserves the terminal object and with the structure of a global higher-order cubical renaming algebra $\MC$ over $F^\ast(\Tele_\MD)$ such that both $\MC.\MI$ and $\MC.\Var$ are locally representable and both $\MC.\mathsf{int}$ and $\MC.\mathsf{var}$ strictly preserve context extensions.
  \lipicsEnd{}
\end{definition}

The category $\CA_\square$ of cubical renamings (or category of atomic cubical contexts and substitutions) is the initial first-order cubical renaming algebra over $\CS_{\mathsf{CTT}}$.
We write $G$ for the functor $G : \CA_\square \to \CS_{\mathsf{CTT}}$, and pose $\CS_G \triangleq G^\ast(\CTele_{\MS_{\mathsf{CTT}}})$.

\begin{theorem}[Induction principle relative to $F : \CA_{\square} \to \CS_{\mathrm{CTT}}$]\label{thm:relative_induction_principle_cubical_renaming}
  Let $\MS^\bullet$ be a displayed higher-order model of CTT over $\CS_G$, and assume
  given the additional data of functions
  \begin{alignat*}{1}
    & \var^{\bullet} : \forall A\ (A^\bullet : \Ty^\bullet(A))\ (a : \Var(A)) \to \Tm^\bullet(A^\bullet, \var(a)), \\
    & \mathsf{int}^{\bullet} : (i : \square.\MI) \to \MI^\bullet(\mathsf{int}(i))
  \end{alignat*}
  such that
  \begin{alignat*}{1}
    & \mathsf{int}^{\bullet}(0) = 0^\bullet, \\
    & \mathsf{int}^{\bullet}(1) = 1^\bullet.
  \end{alignat*}

  Then there exists a section $\sem{-}$ of $\CScone_{\MS^\bullet}$. It satisfies the additional equalities
  \begin{alignat*}{1}
    & \sem{\var_A(a)} = \var^\bullet(\sem{A},a), \\
    & \sem{\mathsf{int}(i)} = \mathsf{int}^\bullet(i).
  \end{alignat*}
\end{theorem}
\begin{proof}
  By the same methods as the proofs of~\cref{thm:relative_induction_principle_renamings} and~\cref{thm:relative_induction_principle_cubical}.
\end{proof}

%%% Local Variables:
%%% mode: latex
%%% TeX-master: "main"
%%% End: